\documentclass[a4paper,12pt]{article}

\usepackage{a4wide}
\usepackage{amssymb}
\usepackage{amsmath}
\usepackage{amsthm}
\usepackage[mathscr]{euscript}
\usepackage{hyperref}

\numberwithin{equation}{section}

\theoremstyle{definition}
\newtheorem{definition}{Definition}[section]

\theoremstyle{plain}
\newtheorem{Theorem}[definition]{Theorem}
\newtheorem*{Theorem411}{Theorem \ref{thm:c1hawk1}}
\newtheorem*{TheoremA1}{Theorem \ref{thm:myers}}
\newtheorem*{Theorem211}{Theorem \ref{thm:stability}}
\newtheorem*{Theorem413}{Theorem \ref{thm:c1hawk2}}
\newtheorem*{Theorem55}{Theorem \ref{thm:c1penrose}}
\newtheorem{Proposition}[definition]{Proposition}
\newtheorem{Lemma}[definition]{Lemma}
\newtheorem{Corollary}[definition]{Corollary}
\newtheorem*{Theorem*}{Theorem}

\theoremstyle{remark}
\newtheorem{remark}[definition]{Remark}

\newcommand{\R}{\mathbb R}
\newcommand{\N}{\mathbb N}

\global\long\def\vol#1{\mathrm{{\mathrm{Vol}\left(#1\right)}}}

\global\long\def\res#1#2{\left.#1\right|_{#2}}
\global\long\def\supp{\mathrm{{supp\,}}}

\newcommand{\eps}{\varepsilon}

\usepackage[xcolor]{changebar}
\cbcolor{blue}
\newcommand{\Ric}{\mathrm{Ric}}

\newcommand{\comp}{\Subset}

\newcommand{\sse}{\subseteq}

\usepackage[inline]{enumitem}
\newcommand{\enumlabelformat}{\roman}
\newcommand{\enumlabelfont}[1]{#1}
\newlength{\thelabelsep}
\setlist{labelsep=\thelabelsep}
\setlist[enumerate]{font=\enumlabelfont,label=(\enumlabelformat*),leftmargin=2.5em}
\setlist[itemize]{leftmargin=2.5em,label=$-$}

\newcounter{inlineenum}
\renewcommand{\theinlineenum}{\enumlabelformat{inlineenum}}

\newcommand\T{\mathcal{T}}
\newcommand\MX{\mathcal{X}}

\newcounter{mnotecount}

\title{Singularity theorems for $C^1$-Lorentzian metrics
	}
\author{Melanie Graf\footnote{\, mgraf2@uw.edu, University of Washington, Department of Mathematics.}
}

\begin{document}

\date{\today}

\maketitle

\begin{abstract}
Continuing recent efforts in extending the classical singularity theorems of General Relativity to low regularity metrics, we give a complete proof of both the Hawking and the Penrose singularity theorem for $C^1$-Lorentzian metrics - a regularity where one still
has existence but not uniqueness for solutions of the geodesic equation. The proofs make use of careful estimates of the curvature of approximating smooth
metrics and certain stability properties of long existence times for causal geodesics. 

On the way we also prove that for globally hyperbolic spacetimes with a $C^1$-metric  causal geodesic completeness is $C^1$-fine stable. This improves a similar older stability result of Beem and Ehrlich where they also used the $C^1$-fine topology to measure closeness but still required smoothness of all metrics.

Lastly, we include a brief appendix where we use some of the same techniques in the Riemannian case to give a proof of the classical Myers Theorem for $C^1$-metrics.  

\vskip 1em


\end{abstract}

 \tableofcontents

\section{Introduction}\label{sec:intro}

The classical singularity theorems of General Relativity show that a Lorentzian
manifold with a $C^2$-metric satisfying certain physically reasonable curvature and causality conditions cannot be causal geodesically complete. In particular, if one attempts to ''extend'' such a manifold to make it complete, one cannot accomplish this with a $C^2$-Lorentzian metric without violating at least one of these physically reasonable conditions. It is then natural to ask whether one could extend with a lower regularity Lorentzian metric - a question that was already raised in \cite[Sec.~8.4]{HE} and is of major physical importance because a plethora of physical models involve spacetime metrics which are not $C^2$ (e.g., the Oppenheimer–Snyder
model of a collapsing star \cite{OS}, matched spacetimes \cite{MS}, self-gravitating compressible fluids \cite{BLF}, gravitational shock waves, thin mass shells, \dots).
 In this paper we show that in the specific cases of both the Hawking theorem, \cite{H:67}, and the Penrose theorem, \cite{Pen}, even $C^1$-regularity of the metric is sufficient to guarantee causal geodesic incompleteness.

In recent years, low regularity analytic methods have become increasingly
important within Mathematical General Relativity, allowing one to phrase physically well-founded intuitions in a framework promising for a rigorous resolution. 
 Looking at the particular issue of extending the singularity theorems to lower regularity metrics, one immediately notices that the parts of the analysis involving the use of Jacobi fields to get the existence of conjugate/focal points become very problematic once the regularity of the metric drops below $C^2$. Further, even some standard results from causality theory employed in the smooth proofs have only recently been investigated for metrics of regularity below $C^2$. This was most successful in the case where the metric is still $C^{1,1}$, see \cite{Ming,KSS,KSSV}, but also lower regularities have been treated (\cite{CG,Clemens}). These advances and the development of approximation techniques adapted to the causal structure allowed the rigorous proof of the Hawking and the Penrose singularity theorem for $C^{1,1}$-metrics (see \cite{hawkingc11,penrosec11}) and, finally, also a proof of the more general Hawking–Penrose Theorem in this regularity could be given, \cite{GGKS:18}.

The next major threshold in lowering the regularity of the metric required for the singularity theorems is $C^1$. Compared with $C^{1,1}$-metrics we lose both uniqueness of geodesics and well-definedness of the curvature in $L^\infty_\mathrm{loc}$ but we retain the following key properties:
\begin{itemize}
\item[(i)] The Levi-Civita connection is continuous. This is, therefore, the lowest regularity
where the classical Peano existence theorem gives existence of solutions
of the geodesic equations.
\item[(ii)] The curvature of the metric is well-defined as a distribution. This allows one to define Ricci curvature bounds. In particular the strong energy condition can be extended to the distributional case in a relatively straightforward fashion (see Definition \ref{def:sec}). As discussed in detail in section \ref{sec:penrose}, the definition of the null energy condition is less clear 
 but a sensible definition can still be made (see Definition \ref{def:genNEC}).
\end{itemize}

Note that point (i) allows us to extend the classical definitions of geodesics and of causal geodesic completeness to $C^1$ metrics: We define geodesics as  solutions of the geodesic equation (see Definition \ref{def:geod}) 
and call a spacetime causal geodesically complete if {\em every inextendible causal geodesic} is complete 
(see Definition~\ref{def:complete}). Due to non-uniqueness this notion of causal geodesic completeness is not the only possible generalization -- one could instead merely demand that there exists a complete causal geodesic for every causal initial data -- but our definition is well in line with physical expectations and still weaker than using, for example, b-completeness\footnote{Defined as all inextendible $C^1$ curves being complete when parametrized with respect to a generalized affine parameter this, at least for all causal curves, has also be argued as being a necessary condition  for  a spacetime to be considered singularity free (cf.~\cite[p.~258-260]{HE}).}.  While most of our results really do require all inextendible timelike or null geodesics to be complete,  Theorem \ref{thm:c1hawk2} also works with the weaker generalization mentioned above (in fact, it would only need the existence of any complete timelike curve, see Remark \ref{rem:complete}). \vspace{10pt}

After going from $C^2$ to $C^{1,1}$ this step down to $C^1$ further reduces the gap between the regularity in which singularity theorems are available and the lower, but physically very relevant regularities obtained from
even classical existence results for the Cauchy problem ($H^s_\mathrm{loc}$ with $s>5/2$, cf., e.g.,\cite{Rendall})
or used in the currently favored formulation of the strong cosmic censorship conjecture (locally square integrable Christoffel symbols, see \cite{christ}, regarded as the lowest regularity where weak solutions of the Einstein equations are well defined).\footnote{In general there seems to be a bit of a mismatch between common regularity classes used when one takes a more causality theoretic ($C^2$, $C^{1,1}$, $C^1$, $C^{0,\alpha}$, \dots ) or a more PDE based approach (Sobolev spaces). Of course these regularities can be compared using Sobolev embedding theorems, e.g., we have $W^{2,p}_\mathrm{loc} \hookrightarrow C^1$ if $p$ is strictly greater than the spacetime dimension.}

 While it would of course be most satisfactory from a physical point of view to have versions of the singularity theorems for these regularities, multiple new mathematical 
 problems arise when attempting to go below $C^1$, one of them being the definition of incompleteness itself.  For metrics that are locally Lipschitz, a regularity where Hawking and Ellis \cite[p.~286-287]{HE} still believe the singularity theorems to hold but no longer give even a heuristic argument, the geodesic equation becomes an ODE with a locally bounded right hand side. So, to use it to define causal geodesics and then causal geodesic completeness one would have to choose one of many different solution concepts available for ODEs with right hand side in $L^\infty_\mathrm{loc}$ (one of the possibilities is using Filippov solutions, \cite{Filippov}, see also \cite{S:14,SS:18}, where this concept is used to define geodesics). Another option would be to try to avoid the geodesic equation entirely and define causal geodesics as locally length maximizing causal curves (see Remark \ref{rem:geoddef}) but this leaves the problem of defining an affine or generalized affine parameter for null geodesics or null curves in general (see also Remark \ref{rem:geodcompletedef}). 
And one problem on the more technical side is that the Friedrichs lemma type estimates in Lemma \ref{lem:Fthesecond} already appear to be quite sharp, so they might fail even when lowering regularities just a little. And once one goes below Lipschitz even pure causality theory departs significantly from classical
theory: Example 1.11 in \cite{CG} shows that the push-up Lemma 
does
not necessarily hold in $C^{0,\alpha }$ spacetimes, $\alpha \in (0,1)$, and that lightcones no longer have to be hypersurfaces (i.e., ''causal bubbling'' may occur). Despite this, there exist a handful of recent inextendibility results where even a very low level of regularity cannot be maintained (e.g.,~\cite{Sbierski, MS:19}). However, these either concern concrete situations with a large amount of symmetry or assume timelike geodesic completeness. 
 
 Lastly we want to mention that a first synthetic singularity theorem for a specific family of Lorentzian length spaces was recently shown in 
 \cite{AGKS:19}. There the vast improvement in the regularity of the metric/spacetime (in fact one does not need a spacetime in the usual sense at all) comes at the expense of requiring a large amount of general structure/symmetry.  \\

The main results of this paper are the following three theorems from sections \ref{sec:hawk} and \ref{sec:penrose}:

\begin{Theorem411}[$C^1$-Hawking, version 1] Let $(M,g)$ be a time-oriented Lorentzian manifold with $g\in C^1$. If $(M,g)$ has non-negative timelike Ricci curvature (in the distributional sense of Def.~\ref{def:sec}) and there exists a smooth compact spacelike hypersurface $\Sigma$ with $g(H,\mathbf{n})>0$ (where $\mathbf{n}$ denotes the future pointing unit normal to $\Sigma$ and $H$ is the mean curvature vector field) 
	on $\Sigma$, then M is not timelike geodesically complete (i.e., there exists at least one incomplete timelike geodesic). 
\end{Theorem411}

\begin{Theorem413}[$C^1$-Hawking, version 2] Let $(M,g)$ be a time-oriented Lorentzian manifold with $g\in C^1$. If $(M,g)$ has non-negative timelike Ricci curvature (in the distributional sense of Def.~\ref{def:sec}) and there exists a smooth spacelike Cauchy hypersurface $\Sigma$ with $(n-1)g(H,\mathbf{n})>\beta >0$ on $\Sigma$, then
$\tau_\Sigma(p)\leq \frac{n-1}{\beta}$ for all $p\in I^+(\Sigma)$.
\end{Theorem413}

\begin{Theorem55}[$C^1$-Penrose] Let $(M,g)$ be a time-oriented globally hyperbolic Lorentzian manifold with $g\in C^1$ and non-compact Cauchy hypersurfaces. If $(M,g)$ satisfies the distributional null energy condition (in the sense of Def.~\ref{def:genNEC}) and contains a closed smooth achronal spacelike $(n-2)$-dimensional submanifold $S$ with past pointing timelike mean curvature vector field $H$, 
then it is null geodesically incomplete. 
\end{Theorem55}

With these results $C^1$ becomes the current lowest regularity for which straight-up generalizations of classical singularity theorems exist. Whether also the more general Hawking-Penrose theorem can be generalized to this regularity is still open.

On the way, in section \ref{sec:geod}, we prove some basic facts about geodesics for $C^1$-metrics, including the following stability result for causal geodesic completeness:

\begin{Theorem211}
Let $(M,g)$ be globally hyperbolic with a $C^1$-metric $g$ which is causal geodesically complete. Then there exists a $C^1$-fine neighborhood $U$ in the space of $C^1$ pseudo-Riemannian metrics on $M$ around $g$ such that all $g'\in U$ are globally hyperbolic Lorentzian metrics that are causal geodesically complete.
\end{Theorem211}

This improves a similar older stability result of Beem and Ehrlich where they also used the $C^1$-fine topology stability to measure closeness but still required smoothness of all metrics.

In section \ref{sec:curv} we review distributions (in the sense of generalized functions) on manifolds and define distributional curvature and curvature bounds for $C^1$-metrics.
 
Lastly, we include a brief appendix, appendix \ref{sec:Riem}, where we use some of the same techniques in the Riemannian case to give a proof of the classical Myers Theorem for $C^1$-metrics: 
 
\begin{TheoremA1}
	Let $(M,h)$ be a Riemannian manifold with a $C^1$-metric $h$. Assume that $h$ is geodesically complete and that there exists a constant $\lambda>0$ such that $\Ric (\MX,\MX)-(n-1)\lambda\, h(\MX,\MX)$ is a non-negative distribution for all smooth vector fields $\MX$ on $M$. Then $\mathrm{diam}(M)\leq \frac{\pi}{\sqrt{\lambda}}$.
\end{TheoremA1}

\subsection*{Notation and Conventions}
All manifolds are assumed to be smooth, Hausdorff, second 
countable, $n$-dimensional (with $n\geq 2$ except in section \ref{sec:penrose}, where $n\geq 3$), and connected. On such manifolds  $M$ we will 
consider Lorentzian (or, as in appendix \ref{sec:Riem}, Riemannian) metrics $g$ that are at least $C^1$. We always assume that $(M,g)$ is time orientable and that its time 
orientation is fixed by a smooth vector field.

We say a curve $\gamma : I\to M$ 
from some interval $I\sse \R$ to $M$ is timelike (causal, null, future or past 
directed) if it is locally Lipschitz and $\dot{\gamma}(t)$, 
which exists almost everywhere by Rademacher's theorem, is timelike (causal, 
null, future or past directed) almost everywhere. Following standard notation, 
for $p,q\in M$ we write $p\ll q$ if there exists a future directed timelike 
curve from $p$ to $q$ (and $p\leq q$ if there exists a future directed causal 
curve from $p$ to $q$ or $p=q$) and set $I^+(A):=\{q\in M:\, p\ll q\ 
\mathrm{for\,some}\,p\in A\}$ and $J^+(A):=\{q\in M:\, p\leq q\ 
\mathrm{for\,some}\,p\in A\}$. We call a spacetime $(M,g)$ 
globally hyperbolic if it is non-total imprisoning (i.e., no future or past inextendible causal curves may be contained in a compact set) and 
$J(p,q):=J^+(p)\cap J^-(q)$ is compact for all $p,q\in M$ (cf.~\cite[Def.~3.1]{Clemens}).

As usual, for $p,q\in M$ the future
time separation from $p$ to $q$ is defined by 
\[
\tau(p,q):=\sup(\left\{ L(\gamma):\gamma\;\text{is a future directed causal curve form }p\text{ to }q\right\} \cup\{0\}),\label{eq:point time sep}
\]
where $L(\gamma)$ denotes the Lorentzian arc-length of $\gamma$, i.e., for a causal curve $\gamma:I\to M$ one has $L(\gamma):=\int_I\sqrt{|g(\dot{\gamma}(t),\dot{\gamma}(t))|}dt$.  Similarly one defines the future time separation to a subset $\Sigma$
by
\[
\tau_{\Sigma}(p):=\sup_{q\in\Sigma}\tau(q,p).\label{eq:subset time sep}
\]

We interchangeably use $\gamma : I\to M$ and $\dot{\gamma}:I\to TM$ to denote geodesics: The first is the standard notation and the second will be useful for proving the convergence and long-time existence results in section \ref{sec:geod}  since it turns the geodesic equation into a first order ODE.

For two Lorentzian metrics $g_1$ and $g_2$ we define
\[g_1\prec g_2 :\iff g_1(X,X)\leq 0 \implies g_2(X,X)<0. \]
We say that $g_1$ has narrower lightcones than $g_2$ (or alternatively that $g_2$ has wider lightcones than $g_1$).

Next we briefly discuss two low-regularity specific concerns, namely the definition of global hyperbolicity and the regularity assumptions on causal curves.

\begin{remark}[Regarding the definition of global hyperbolicity]\label{rem:globhypdef} Following \cite[Def.~3.1]{Clemens} we define globally hyperbolic using non-total imprisoning (i.e., no future or past inextendible causal curves may be contained in a compact set) and
compactness of causal diamonds $J(p,q)=J^+(p)\cap J^-(q)$. It is shown in \cite[Proposition~2.20]{Ming:19} that for proper cone structures (and thus in particular for continuous metrics) this definition is equivalent to demanding causality and compactness of causal diamonds. Further, \cite{HM:19} shows that for spacetimes of dimension greater than two with smooth or at least $C^{1,1}$-metrics one may even replace causality with the mere assumption of the spacetime being either non-totally vicious or non-compact.

These definitions are equivalent to the existence of a Cauchy hypersurface (i.e., a set that is met
exactly once by every inextendible causal curve), see \cite[Theorem~5.7 and 5.9]{Clemens}. From this it is clear that if $g$ is globally hyperbolic, then any $g'\prec g$ (remember that this was defined as $g'(X,X)\leq 0\implies g(X,X)<0$) must be globally hyperbolic too.
\end{remark}

\begin{remark}[Regarding the regularity of causal curves]\label{rem:regofcurves}
Note that in line with, e.g.,~\cite{CG, GGKS:18, Clemens, GKSS:19} we only require causal and timelike
curves to be Lipschitz, whereas classical accounts use smooth or piecewise smooth or piecewise $C^1$ 
curves instead (see, e.g., \cite{HE, ON83}). Another convention that is sometimes used is requiring timelike curves to be piecewise $C^1$ while causal curves only have to be Lipschitz (e.g.,~\cite{Ming:19}). For a general discussion of the different conventions and their relations see also \cite{GKSS:19}, Section 2.

 In our case it does not matter whether one uses smooth or Lipschitz timelike curves for the definition of the relation $\ll$ because $g\in C^1$ and hence in particular causally plain (cf.~\cite{CG}, \cite[Theorem 2.15]{GKSS:19}). For the relation $\leq$ it will follow from Proposition~\ref{prop:existmaxgeods} that $J^+(p)\equiv J^+_{\mathrm{Lip}}(p)=J^+_{C^1}(p)$ if $(M,g)$ is globally hyperbolic (when using Lipschitz causal curves). Proposition~\ref{prop:existmaxgeods} further implies that for globally hyperbolic spacetimes also the definition of the time separation function does not depend on whether one uses Lipschitz or $C^1$ causal curves.
\end{remark}

We use $\pi :TM \to M$ to denote the canonical projection and
 we also fix a complete smooth Riemannian background metric on $M$, denoted by $h$. Note that nothing we are going to do will depend on the choice of $h$, we nevertheless fix $h$ for simplicity.
  By a slight abuse of notation we also write $h$ for the complete Riemannian metrics on $TM$ and $T^*M$ induced by $h$. For any continuous $(r,s)$-tensor field $\T$ on $M$ and any subset $A\sse M$ we define
 \begin{align*}||\T||_{\infty,A}:=&\sup \{|\T|_x(\omega_1,\dots,\omega_r,X_1,\dots,X_s)|:x\in A,
\omega_i\in T_x^*M, ||\omega||_h=1,\\
&X_j\in T_xM,||X_j||_h=1\}.
 \end{align*}
If $A\sse M$ is compact, then a net $\{\T_\eps\}$ of continuous $(r,s)$-tensor fields converges to $\T$ w.r.t.~$||.||_{\infty,A}$ if and only if $(\T_\eps^\alpha)^{i_1,\dots,i_r}_{j_1,\dots,j_s} \to (\T^\alpha)^{i_1,\dots,i_r}_{j_1,\dots,j_s} $ uniformly on $U_\alpha \cap A$ for any chart $(\psi_\alpha,U_\alpha)$. If $\T _\eps \to \T$ w.r.t.~$||.||_{\infty,A}$ for any compact $A\sse M$ we say that $\T_\eps \to \T$ locally uniformly (or \emph{in $C^0_\mathrm{loc}$}). 

If $\T,\T_\eps$ are $C^k$, we say $\T_\eps \to \T$ in $C^k_\mathrm{loc}$ if both $\T_\eps \to \T$ and $(^h\nabla)^i \,\T_\eps \to \,(^h\nabla)^i \,\T$ in $C^0_\mathrm{loc}$ for all $1\leq i\leq k$ where $^h\nabla $ denotes the Levi-Civita connection of the background metric $h$. Note that convergence of $^h\nabla \T_\eps \to\, ^h\nabla \T$ in $C^0_\mathrm{loc}$ is equivalent to locally uniform convergence of all partial derivatives $\partial_{l}(\T_\eps^\alpha)^{i_1,\dots,i_r}_{j_1,\dots,j_s} \to \partial_{l}(\T^\alpha)^{i_1,\dots,i_r}_{j_1,\dots,j_s} $ for any chart $(\psi_\alpha,U_\alpha)$. Further, for $C^1$-metrics the Levi-Civita connection $\nabla \equiv \,^g\nabla$ is well defined (e.g.~via the Koszul formula) and (as long as $\T, \T_\eps$ are at least $C^1$) one has that $\nabla \T, \nabla \T_\eps$ are a continuous $(r,s+1)$-tensor fields and $^h\nabla \T_\eps \to\, ^h\nabla \T$ in $C^0_\mathrm{loc}$ if and only if $\nabla \T_\eps \to \nabla \T$ in $C^0_\mathrm{loc}$. 

\section{Geodesics for $C^1$-metrics}\label{sec:geod}

Since $g\in C^1$ the Christoffel symbols are continuous and we can make the following definition:

\begin{definition}[Geodesic]\label{def:geod} Let $(M,g)$ be a spacetime with a $C^1$-metric $g$. A $C^2$ curve $\gamma :I\to M$ is a \emph{geodesic} if it satisfies $\nabla_{\dot{\gamma}}\dot{\gamma}=0$, i.e., in coordinates $\ddot{\gamma}^k=-\Gamma^k_{ij} \circ \gamma \; \dot{\gamma}^i\dot{\gamma}^j$. It is called a \emph{pre-geodesic} if there exists a continuous function $\lambda :I\to \R$ such that $\nabla_{\dot{\gamma}}\dot{\gamma}=\lambda \dot{\gamma}$.

A geodesic $\gamma :I\to M$ is called \emph{extendible} if there exists $J\supsetneq I$ and a geodesic $\tilde{\gamma} :J\to M$ such that $\tilde{\gamma}|_{I}=\gamma $. It is called \emph{inextendible} or \emph{maximal}\footnote{Not to be confused with the separate notion of being (length-)maximizing.} if it is not extendible.
\end{definition}

\begin{remark}[Geodesics versus locally length maximizing causal curves]\label{rem:geoddef} Up to parametrization there are two different equivalent descriptions of causal  geodesics for smooth (and even $C^{1,1}$, cf.~\cite[Theorem~6]{Ming}) metrics: Either as above as causal curves solving the geodesic equation or as locally length maximizing causal curves. While in this paper ''geodesic'' will exclusively refer to curves as in Definition~\ref{def:geod} the second description is also commonly used to define causal geodesics in low-regularity Lorentzian geometry. This has the advantage of being a well-defined concept even in very low regularities but ignores the issue of parametrization. In the timelike case the correct parametrization can easily be determined: For smooth metrics a locally length maximizing timelike curve parametrized with respect to Lorentzian arc-length is $C^2$ and solves the geodesic equation. It is unknown whether this continues to be true if the metric is only $C^1$ and even whether locally maximizing causal curves posses a parametrization improving their regularity (existence results generally show only the existence of {\em Lipschitz continuous} causal maximizers, cf.~e.g., \cite[Proposition 6.4]{Clemens}). The converse, i.e., any timelike solution of the geodesic equation being locally length maximizing, fails already for metrics in $C^{1,\alpha}$ for any $\alpha <1$ (see \cite[Ex.~2.4]{SS:18}). 
\end{remark}

Note that for any initial data $\dot{\gamma}(0) \in TM $ there exist geodesics, but one does in general not have uniqueness. This is immediate from the Christoffel symbols being continuous but not necessarily Lipschitz and explicit examples of non-uniqueness can be found e.g. in \cite[Ex.~2.4]{SS:18} (cf.~also \cite[Sec.~5]{HW:51} for the Riemannian case). Despite possible issues with non-uniqueness lots of the usual basic properties remain true: A geodesic is inextendible if and only if $\dot{\gamma}$ leaves every compact subset of $TM$, any geodesic has fixed causal character and a causal geodesic $\gamma :(a,b)\to M$ is extendible as a geodesic if and only if it can be extended as a continuous curve. The first of these facts is basic continuous ODE theory (note that boundedness of $\dot{\gamma}$ implies that the end of $\gamma $ lies in single chart, so this is just \cite[Thm 3.1]{Hartman}), the other two will be shown using the next lemma.

\begin{Lemma}\label{lem:blah}
	Let $(M,g)$ be a spacetime with a locally Lipschitz metric. If $g$ is $C^1$ on some open subset $O\subset M$ and $\gamma :[0,1)\to O$ is a causal geodesic that is continuously extendible to $p:=\gamma(1)\in \partial O$ then $\dot{\gamma}$ remains bounded on $[0,1)$. Further, $\dot{\gamma}(t)$ cannot converge to zero as $t\to 1$.
\end{Lemma}
\begin{proof}
	W.l.o.g.~$\gamma $ is future directed and contained in a single chart around $p$ and in this chart $-dt^2+4 \sum_i dx_i^2\prec g \prec -dt^2+\frac{1}{4} \sum_i dx_i^2$. Then since $\gamma $ is causal, we estimate $|\dot{\gamma}^i|\leq 2 |\dot{\gamma}^0|$ for $n-1\geq i\geq 1$ on $[0,1)$, i.e., the function $t\mapsto f^i(t):=\frac{\dot{\gamma}^i(t)}{\dot{\gamma}^0(t)}$ is bounded on $[0,1)$. By the geodesic equation
	\[\ddot{\gamma}^0(t)=\big[-\Gamma^0_{00}(\gamma(t)) -2 \sum_{i=1}^{n-1} f^i(t) \Gamma^0_{0i}(\gamma(t)) -\sum_{i,j=1}^{n-1} f^i(t) f^j(t) \Gamma^0_{ji}(\gamma(t))\big] (\dot{\gamma}^0(t))^2=:f(t) (\dot{\gamma}^0(t))^2\]
	for $t\in[0,1)$ and $f$ is continuous and bounded on $[0,1)$. Integrating this gives $\dot{\gamma}^0(t)=\frac{1}{c_0-\int_0^tf(s) ds}$, where $c_0=\frac{1}{\dot{\gamma}^0(0)}$. This cannot converge to zero and diverges if and only if $\int_0^1f(s) ds=c_0$. But then $|c_0-\int_0^tf(s) ds|\leq C |t-1|$ because the function $t\mapsto c_0-\int_0^tf(s) ds$ is Lipschitz on $[0,1]$. So $\dot{\gamma}^0(t)=|\dot{\gamma}^0(t)|\geq \frac{1}{C|t-1|}$, contradicting integrability of $\dot{\gamma}^0$ on $[0,1)$.
\end{proof} 

So the velocity $\dot{\gamma}$ of any causal geodesic $\gamma$ that is continuously extendible will remain in a compact set of $TM$, hence $\gamma $ will be extendible as a geodesic. 

\begin{Corollary}Geodesics have fixed causal character. \end{Corollary}
\begin{proof}
We first note that, as for smooth metrics, $g(\dot{\gamma},\dot{\gamma})$ is constant along geodesics. This leaves only one possible scenario in which the causal character of $\gamma $ changes, namely that $\dot{\gamma}$ might switch from being null (and non-zero) to being zero (which is spacelike). That this cannot happen follows from $\lim_{t\to a} \dot{\gamma}(t)\neq 0$ for any geodesic $\gamma :[0,a]\to M$ which is causal on $[0,a)$ (cf.~Lemma \ref{lem:blah}).
\end{proof}

Finally we give our definition of causal/timelike/null geodesic completeness.

\begin{definition}[Geodesic completeness]\label{def:complete} A spacetime $(M,g)$ with a $C^1$-metric $g$ is called timelike (respectively null) geodesically complete if \emph{all} inextendible timelike (respectively null) solutions of the geodesic equation are defined on $\R$. 
\end{definition}

\begin{remark}[Causal (in-)completeness for metrics below $C^1$]\label{rem:geodcompletedef}

By defining timelike geodesics as locally maximizing timelike curves and an inextendible timelike curve to be complete if it has infinite length, we may extend a definition of timelike geodesic completeness to regularities below $C^1$.\footnote{Note, however, that for $C^1$ metrics this version of timelike geodesic completeness doesn't imply the one given in Definition \ref{def:complete} and it is unknown whether the converse holds (cf.~Remark~\ref{rem:geoddef}).} This is a very robust definition and still works well for continuous metrics (and even Lorentzian length spaces).

However, the situation for null curves/null geodesics is a lot more problematic and it is not yet clear what a good definition of null geodesic completeness would look like even for Lipschitz metrics: While we can still define null geodesics as locally maximizing null curves, their length will always be zero, leaving us with the problem of finding the right parametrization. And unfortunately (without further restrictions on the differentiability of the metric, like, e.g.,~$g$ being smooth away from a spacelike hypersurface) the standard definition of a generalized affine parameter (as is used classically to define b-completeness, \cite[p.~259]{HE}) might fail completely: Since the metric is only differentiable almost everywhere the Christoffel symbols might not be defined at all along any given null curve, making the classical definition of parallel transport along this curve impossible. This leaves us with coming up with a new way to define parallel transport or trying to use the geodesic equation itself, which would mean dealing with a nonlinear ODE with an $L^{\infty}_{\mathrm{loc}}$ right hand side and while the concept of Filippov solutions, \cite{Filippov}, has already been used to some success, see \cite{S:14,SS:18}, it is unknown whether this gives a useful completeness notion.
\end{remark}

\subsection{Existence time estimates for $g_\eps$-geodesics}

In this section we will consider a net of $C^1$-Lorentzian metrics $\{g_\eps\}_{\eps>0}$ converging to a $C^1$-Lorentzian metric $g$ in $C^1_\mathrm{loc}$ and study what timelike (or null) completeness of $g$ implies for the existence time of geodesics for the $g_{\eps}$-metrics. This will be later applied to the situation where the $g_\eps$ are actually smooth approximations.

The main goal of this section is to establish Proposition~\ref{prop:relcompactfunnels}, stating that $g_\eps$-geodesics must have arbitrarily long existence times if corresponding $g$-geodesics are complete. The strategy of the proof is similar to \cite[Theorem~5]{Filippov}, which shows that, even for ODE with a right-hand side having certain mild discontinuities, if all solutions for given initial data exist on $[a,b]$, then the set of points lying on the graphs of these solutions is bounded and closed. 
 
 We will split the argument into several Lemmas for easier reference later on. To start with we note that the following holds:
	\begin{Lemma}\label{lem:aa} Let $\{\dot{\gamma}_n:[0,a_n]\to TM\}_{n\in \N}$ be a sequence of $g_{\eps_n}$-geodesics with $\eps_n\to 0$ and assume there exists a compact set $B\sse TM$ such that $\mathrm{im}(\dot{\gamma}_n)\sse B $, $\dot{\gamma}_n(a_n)\in \partial B$ for all $n$ and $\dot{\gamma}_n(0)\to v\in B^\circ $. Then there exists $0<a\leq \infty $ and a $g$-geodesic $\dot{\gamma}_v :[0,a)\to B\sse TM$ such that a subsequence of the $\dot{\gamma}_n$'s converges uniformly on compact subsets of $[0,a)$ to $\dot{\gamma}_v$.
	\end{Lemma}
	\begin{proof} This is basically a consequence of Arzela-Ascoli. W.l.o.g.~$\dot{\gamma}_n(a_n)\to q\in \partial B$, $a_n\to a \in [0,\infty]$ (by compactness of $\partial B$ there exists a convergent subsequence $\dot{\gamma}_{n_k}(a_{n_k})$, so if $\dot{\gamma}_n(a_n)$ doesn't converge and/or $a_n$ has multiple accumulation points for the the original sequence $\gamma_n$, we replace it with a suitable subsequence). First we argue that $a>0$: Since $q\in \partial B$ and $v\in B^\circ$, there exists $c>0$ such that $d_h(\dot{\gamma}_n(0),\dot{\gamma}_n(a_n))\geq c$ for $n$ sufficiently large, so if $a_n\to 0$, then $\ddot{\gamma}_n$ cannot remain bounded (uniformly in $n$), contradicting the geodesic equation since $\dot{\gamma}_n \sse B$ and $^{\eps_n} \Gamma \to \Gamma \in \mathcal{C}^0$ uniformly on the compact set $\pi(B)\sse M$.
	
	The sequence $\dot{\gamma}_n$ is uniformly bounded and by the geodesic equation and convergence of the $^{\eps_n}\Gamma$'s the same is true for the sequence of derivatives $\ddot{\gamma}_n$, hence $\dot{\gamma}_n$ is equicontinuous. Thus, by Arzela-Ascoli, there exists a subsequence converging uniformly on compact subsets of $[0,a)$ to a function $f:[0,a)\to TM$.  By the geodesic equation also the (sub)sequences of derivatives $\ddot{\gamma}^i_n$ converge uniformly on compact sets to functions $g^i=-\Gamma_{jk}^i(\pi(f)) f^j f^k$. This implies that $f$ is differentiable and $\dot{f}=g$, hence $f$ is a solution of the $g$-geodesic equation with initial data $f(0)=\lim \dot{\gamma}_n(0)=v$. 
	\end{proof}

	\begin{Lemma}\label{lem:ab} Let $\{\dot{\gamma}_n:[0,t_n)\to TM\}_{n\in \N}$ be a sequence of inextendible $g_{\eps_n}$-geodesics with $\eps_n\to 0$, $t_n\leq N$ and $\dot{\gamma}_n(0)\to v\in TM$. Then there exists an inextendible $g$-geodesic $\dot{\gamma}:[0,b)\to TM$ with $\dot{\gamma}(0)=v$ and $b\leq N$. Further, a subsequence converges uniformly on compact subsets of $[0,b)$ to this solution $\dot{\gamma}$.
	\end{Lemma}
	\begin{proof}
		Let $\{K_m\}_{m\in\N}$ be a compact exhaustion of $TM$ such that $v\in K^\circ_1$. W.l.o.g.~$v_n\in K_1$ for all $n$. Then, since $\dot{\gamma}_n(s)$ leaves every compact subset of $TM$ as $s\to t_n$, for all $(n,m)$ there exists $s(n,m)<t_n\leq N$ such that $q_{n,m}:=\dot{\gamma}_n(s(n,m))\in \partial K_m$ and $\dot{\gamma}_n|_{[0,s(n,m))}\sse K^\circ_m$.
		
		By the previous Lemma \ref{lem:aa} we can choose a subsequence $\{\dot{\gamma}_{\phi_2(k)}\}_{k\in \N}$ of $\{\dot{\gamma}_n\}$ such that $s(2):=\lim_{k\to \infty}s(\phi_2(k),2)$ and $s(1):=\lim_{k\to \infty}s(\phi_2(k),1)$ exist and $\{ \dot{\gamma}_{\phi_2(k)}|_{[0,s(\phi_2(k),2))}\}_{k\in \N}$ converges (uniformly on compact subsets) to a $g$-geodesic $\dot{\eta}_2:[0,s(2))\to TM$. For this solution $\dot{\eta}_{2}(s(1))\in \partial K_1$. Next, we choose a subsequence $\{ \dot{\gamma}_{\phi_3(k)}\}$ of $\{\dot{\gamma}_{\phi_2(k)}\}$ such that $s(3):=\lim_{k\to \infty}s(\phi_3(k),3)$ exists and $\dot{\gamma}_{\phi_3(k)}$ converges (uniformly on compact subsets) to a $g$-geodesic $\dot{\eta}_3:[0,s(3))\to TM$. For this solution, $\dot{\eta}_3(s(2))\in \partial K_2$ (and $\dot{\eta}_3(s(1))\in \partial K_1$). Additionally, $\dot{\eta}_3|_{[0,s(2))}=\dot{\eta}_2$. 
		
		Continuing this construction, we obtain sequences $\{\{\dot{\gamma}_{\phi_i(k)}\}_{k\in \N}\}_{i\in\N}$ and $\{s(i)\}_{i\in\N}$. The 
		sequence $s(i)$ is increasing and bounded above by $N$, hence $s(i)\to b \in (0,N]$. By construction the diagonal sequence $\{\dot{\gamma}_{\phi_k(k)}\}_{k\in \N}$ converges to a $g$-geodesic $\dot{\eta}:[0,b)\to TM$ with $\dot{\eta}(s(k))\in \partial K_k$ for all $k\in \N$, i.e., $\dot{\eta}(s)\to \infty$ as $s\to b\leq N$. 
	\end{proof}
	
The desired result is now an easy consequence:

\begin{Proposition} \label{prop:longexistence} Let $g,g_\eps \in C^1$, $g_\eps \to g$ in $C^1_\mathrm{loc}$ and let and $K\sse TM$ be compact and assume that all $g$-geodesics starting in $K$ are defined on $[0,\infty)$. Then for any $N\in \N$ there exists $\eps_0(N,K)$ such that for all $\eps \leq \eps_0(K, N)$ all $g_\eps$-geodesics starting in $K$ are defined on $[0,N]$.
\end{Proposition}
\begin{proof}
	Assume not. Then there exists $N\in \N$, a sequence $\eps_n \to 0$  and inextendible $g_{\eps_n}$-geodesics $\dot{\gamma}_n: [0,t_n)\to TM$ with $t_n\leq N$ and $\dot{\gamma}_n(0)=v_n \in K$. W.l.o.g.~$v_n\to v\in K$. Then, by the previous Lemma \ref{lem:ab}, there exists an inextendible $g$-geodesic $\gamma :[0,b)\to M$ with $b\leq N$ and $\dot{\gamma}(0)=v\in K$, a contradiction.
\end{proof}
\begin{remark} \label{rem:longexistence} Clearly, the following modified version of Proposition~\ref{prop:longexistence} holds:
	
	\emph{Let $g\in C^1$, $K\sse TM$ compact and $T\geq 0$ and assume that all $g$-geodesics starting in $K$ are defined on $[0,T)$. Then for any $\tilde{T}<T$ there exists $\eps_0(K, \tilde{T})$ such that for all $\eps \leq \eps_0(K,\tilde{T})$ all $g_\eps$-geodesics starting in $K$ are defined on $[0,\tilde{T}]$.}\\ 
\end{remark}

To state and prove our final result in this section, we first need to introduce some more notation: For a compact $K\subset TM$, $N>0$ we denote the union of the images of all $g$-geodesics $\dot{\gamma}$ (restricted to $[0,N]$ if they exist longer) with $\dot{\gamma}(0)\in K$ by $F_{K,N}$, i.e.,
\begin{equation}\label{eq:F}
 F_{K,N}:=\bigcup_{\{ \dot{\gamma}:\; \mathrm{g-geodesic \; with}\; \dot{\gamma}(0)\in K \}}\mathrm{im}(\dot{\gamma}|_{[0,N]})\sse TM.
\end{equation}
We define $F_{\eps,K,N}$ the same way, only using $g_\eps$-geodesics instead of $g$-geodesics.

\begin{Proposition}\label{prop:relcompactfunnels} Let $g\equiv g_0\in C^1$ and $\eps \mapsto g_\eps$ be continuous from $[0,1]$ to the space of $C^1$-metrics with respect to $C^1_\mathrm{loc}$-convergence. Assume that $K\sse TM$ is compact and that all $g$-geodesics starting in $K$ are defined on $[0,\infty)$. Then for any $N\in \N$ the union \[F_{\leq \eps_0,K,N}:=\bigcup_{0<\eps\leq \eps_0(K,N)} F_{\eps, K,N} \cup F_{K,N}\] (where $\eps_0(K,N)$ is given by Proposition \ref{prop:longexistence}) is relatively compact. 
\end{Proposition}
\begin{proof} If not, there exist points $p_n=\dot{\gamma}_{\eps_n}(t_n)\in F_{\leq\eps_0,K,N}$ with $p_n\to 
	\infty $. W.l.o.g.~$\eps_n\to \bar{\eps}\in [0,\eps_0(N,K)]$, $t_n\to t_0\in[0,N]$ and $\dot{\gamma}_{\eps_n}(0)\to v\in K$. 
	Since $\eps \mapsto g_{\eps}$ was assumed to be continuous we have that $g_\eps \to g_{\bar{\eps}}$ in $C^1_\mathrm{loc}$ as $\eps \to \bar{\eps}$ and we may 
	proceed as in the proof of Lemma \ref{lem:ab}: Since $p_n\to \infty$ we can assume $\dot{\gamma}_{\eps_n}|_{[0,t_n]}\cap \partial K_m \neq \emptyset$ for all $n\geq m$ (passing to a subsequence if necessary). Let $q_{n,m}$ and $s(n,m)\leq t_n \leq N$ be as in  Lemma \ref{lem:ab} (well defined for $n\geq m$). 
	Then as in the proof of Lemma \ref{lem:ab} one obtains a solution $\dot{\eta}:[0,s_0)\to TM$ of the $g_{\bar{\eps}}$-geodesic equation with $s_0=\lim_{m\to \infty}\lim_{n\to\infty}s(n,m)\leq N$ and $\dot{\eta}(s(k))\in \partial K_k$, i.e., $\dot{\eta}(s(k))\to \infty$. This contradicts Proposition \ref{prop:longexistence}, since $\bar{\eps}\leq \eps_0(N,K)$.
\end{proof}

\begin{remark}
	\begin{itemize}
		\item[(i)]\label{rem:compactfunnels}	It is also not hard to see that
		 $F_{\leq\eps_0,K,N}$ 
		is actually compact but we will not need this.
		\item[(ii)] \label{rem:relcompactfunnels} As for Proposition~\ref{prop:longexistence}, the following modified version of Proposition~\ref{prop:relcompactfunnels} holds:
		
			\emph{Let $g\in C^1$ and $K\sse TM$ compact, $T\geq 0$ and assume that all $g$-geodesics starting in $K$ are defined on $[0,T)$. Then for any $\tilde{T}<T$ the union \[F_{\leq \eps_0,K,\tilde{T}}:=\bigcup_{0<\eps\leq \eps_0(\tilde{T},K)} F_{\eps, K,\tilde{T}} \cup F_{K,\tilde{T}}\] (where $\eps_0(K,\tilde{T})$ is as in  Rem.~\ref{rem:longexistence}) is relatively compact.}
	\end{itemize}
\end{remark}

\subsection{Existence of maximizing geodesics}\label{sec:existofmax}
We now turn to discussing the existence of length maximizing geodesics for globally hyperbolic $C^1$-metrics. It is well known that for smooth metrics any length maximizing causal curve must be a geodesic (cf. e.g.~\cite[Proposition~14.19]{ON83}). Further, from \cite{Clemens}, we know that maximizing causal curves in globally hyperbolic spacetimes always exist, even if the metric is merely continuous. If $g\in C^{1,1}$ it is still true that any maximizing causal curve is a geodesic (cf. \cite[Theorem 6]{Ming}) and thus maximizing geodesics exist. 

Now for $C^1$-metrics it is unknown whether maximizing curves have to be (reparametrizations of) geodesics and it is even unknown whether one can choose a parametrization in which they are at least $C^2$-curves.

While the arguments below don't answer either of these questions for \emph{all} maximizing causal curves, they at least establish the existence of maximizing curves that are geodesics between any two causally related points if $g$ is globally hyperbolic.

\begin{Proposition}\label{prop:existmaxgeods}
	Let $(M,g)$ be globally hyperbolic with a $C^1$-metric $g$. Then for any $p< q$ there exists at least one maximizing causal \emph{geodesic} from $p$ to $q$. 
	
	Moreover, for any net $\check{g}_\eps$ of approximating metrics converging to $g$ in $C^1_{\mathrm{loc}}$ with $\check{g}_\eps \prec g$ there always exists at least one $L_g$-maximizing $g$-causal $g$-geodesic $\gamma :[0,1]\to M$ from $p$ to $q$ that is obtained as a $C^1$-limit of a sequence of $L_{\check{g}_{\eps_n}}$-maximizing $\check{g}_{\eps_n}$-causal $\check{g}_{\eps_n}$-geodesics $\gamma_{\eps_n}:[0,1]\to M$ with $\eps_n\to 0$.	
\end{Proposition}
\begin{proof}
	We first show that there always exists a timelike maximizing geodesic from $p$ to $q$ for $p\ll q$. First fix a smooth globally hyperbolic metric $\hat{g}$ with $g\prec \hat{g}$ (such a metric exists by \cite[Theorem~4.5]{Clemens}). Further, note that there exists a sequence of metrics $\hat{g}_n$ with $g \prec \hat{g}_n\prec \hat{g}$ and $\hat{g}_n \to g$ in $C^1_{\mathrm{loc}}$ as $n\to \infty$ (this follows essentially from \cite[Proposition 1.2]{CG} and globalizing \cite[Lemma~1.4]{Clemens} using \cite[Lemma 4.3]{HKS:12}, see also the proof of Lemma \ref{lem:approxnice}).
	
	Since the $\hat{g}_n$ have narrower lightcones than $\hat{g}$, they must all be globally hyperbolic as well and $J_{\hat{g}_n} (p,q)\subseteq J_{\hat{g}}(p,q) \comp M$. Let $\gamma_n:[0,1]\to M$ be a $\hat{g}_n$-maximizing $\hat{g}_n$-timelike $\hat{g}_n$-geodesic from $p$ to $q$ (these exist because $\hat{g}_n$ is globally hyperbolic and $p\ll_{\hat{g}_n} q$). 
	
	Let $K_i$ be a compact exhaustion of $TM$
	and let $a_{i,n}$  be the maximum of $s\in [0,1] $ such that $\dot{\gamma}_n|_{[0,s]}\subset K_i$. 	If there does not exist $N$ such that $a_{N,n}=1$ for all large enough $n$, then as in Lemma \ref{lem:ab}
	 we can find a subsequence $\dot{\gamma}_{n_k}$ converging uniformly on compact subsets of $[0,a)$, where $1\geq a=\lim_{i\to\infty}\lim_{k\to \infty}a_{i,n_k}$, to a $g$-geodesic $\dot{\gamma} :[0,a)\to M$. This $\gamma $ must be causal and $\dot{\gamma}(s)$ must leave every compact set as $s\to a$. By Lemma~\ref{lem:blah} for causal geodesics inextendibility as a geodesic is equivalent to inextendibility as a continuous curve, so the causal curve  $\gamma :[0,a)\to M $ is inextendible. But since all curves $\gamma_{n_k}$ are contained in the common compact set $J_{\hat{g}}(p,q)$ the same must be true for the limit  $\gamma $, so $\gamma $ is an inextendible causal curve imprisoned in a compact set. This contradicts global hyperbolicity.	Hence, for large $n$, all curves $\dot{\gamma}_n$ are contained in a common compact subset of $TM$. This implies that a subsequence $\dot{\gamma}_{n_k}$ converges uniformly to a $g$-geodesic $\dot{\gamma}:[0,1]\to M$.

	It remains to show that $\gamma $ is maximizing for $g$. This actually immediately follows from \cite[Proposition~6.5]{Clemens}, but parts of the proof simplify due to the better convergence of $\gamma_{n_k}\to \gamma $ and for completeness we include an outline of the proof here. 
Let $c:[0,1]\to M$ be any $g$-causal curve from $p$ to $q$ parametrized proportional to $h$-arclength. Any such $c$ is also $\hat{g}_n$-causal and, by \cite[Lemma 2.7]{Clemens}, there exists a constant $C$ such that $h(\dot{c},\dot{c})\leq C$ for any such curve $c$. Then as in  \cite[Theorem~6.3]{Clemens} one can estimate 
\[L_g(c) \leq \int_0^1 \sqrt{-\hat{g}_n(\dot{c},\dot{c})+C^2\delta_{n}}\leq 
L_{\hat{g}_n}(c)+C\sqrt{\delta_n}\leq \tau_n(p,q)+C\sqrt{\delta_n},\]
	where $\delta_n$ measures the distance between $g$ and $\hat{g}_n$.  Letting $n \to \infty$, we obtain $L_g(c)\leq \liminf_{n\to\infty}\tau_n(p,q)$ and hence $\tau(p,q)\leq \liminf_{n\to\infty}\tau_n(p,q)$.
On the other hand, the $C^1([0,1])$-convergence of $\gamma_{n_k}$ to $\gamma$ implies that $\tau_{n_k}(p,q)=L_{\hat{g}_{n_k}}(\gamma_{n_k})\to L_g(\gamma)$, hence $\tau(p,q)\leq\lim \tau_{n_k}(p,q) =L_g(\gamma)$, i.e., $\gamma$ is maximizing and since $\tau(p,q)>0$ it is timelike.

 Now let  $\check{g}_\eps$ be a net of approximating metrics converging to $g$ in $C^1_{\mathrm{loc}}$ with $\check{g}_\eps \prec g$. We show that for $p\ll q$ there exists a maximizing timelike $g$-geodesic from $p$ to $q$ that is obtained as the $C^1$-limit of a sequence of $L_{\check{g}_{\eps_k}}$-maximizing $\check{g}_{\eps_k}$-causal $\check{g}_{\eps_k}$-geodesics. Since $J(p,q)$ is compact and all our curves lie in $J(p,q)$ we may use \cite[Lemma~1.4]{Clemens} to assume without loss of generality that $\check{g}_{\eps_2}\prec \check{g}_{\eps_1}$ for $\eps_2<\eps_1$. Let $\eps_0$ be such that $p\ll_{\check{g}_{\eps_0}} q$ (such an $\eps_0$ exists because $p\ll q$ and $I^+=\bigcup_\eps I^+_{\check{g}_\eps}$, see \cite[Corollary 1.17 and Proposition 1.21]{CG}).
Now for all $\eps<\eps_0$, let $\gamma_{\eps}:[0,1]\to M$ be a $\check{g}_\eps$-maximizing $\check{g}_\eps$-timelike $\check{g}_\eps$-geodesic from $p$ to $q$ (these exist because $\check{g}_\eps$ is globally hyperbolic and the width of the $\check{g}_\eps$-lightcones is increasing as $\eps$ decreases). Proceeding exactly the same way as above, we obtain that a subsequence  $\gamma_{\eps_k}$ converges to a $g$-geodesic $\gamma :[0,1]\to M$ in $C^1([0,1])$. It remains to show that $\gamma$ is maximizing. Let $c$ be a timelike $g$-geodesic with $L_g(c)=\tau(p,q)$. Since $c$ is timelike and $C^1$ there exists a constant $\tilde{C}>0$ such that $g(\dot{c},\dot{c})<-\tilde{C}$ and hence $c$ is $\check{g}_\eps$-timelike for all small enough $\eps$. Therefore, for small enough $\eps$, we can estimate $\tau(p,q)=L_g(c) \leq \tau_\eps(p,q)+C\sqrt{\delta_\eps}$
 as above. From this it again follows that $\gamma $ is a maximizing timelike $g$-geodesic.

Lastly, we treat the case $p<q$ but not $\ll q$: Choose a sequence $\eps_k \to 0$ and
points $q_k \in \partial J^+_{\check{g}_{\eps_k}}(p)$ converging to $q$ and let $\gamma_{\eps_k}:[0,1]\to M$ be a $\check{g}_{\eps_k}$-maximizing $\check{g}_{\eps_k}$-null $\check{g}_{\eps_k}$-geodesic from $p$ to $q_k$.\footnote{~Such $q_k$ exist by the same argument as in the proof of \cite[Theorem~5.3]{GGKS:18}: Let $U_k$ be a sequence of neighborhoods of $q$ with $U_{k+1} \sse U_k$ and $\bigcap_k U_k =\{q\}$. Then for any $U_k$ there exist points $q_k^e \in U_k\setminus \overline{J^+(p)} $ and $q_k^i \in U_k \cap I^+(p)$. Let $\eps_k $ be such that $q_k^i \in I^+_{\check{g}_{\eps_k}}(p)$ and $\eps_k \leq \frac{1}{k}$ and let $c_k $ be a curve in $U_k$ connecting $q_k^i$ and $q_k^e\in U_k\setminus \overline{J^+(p)} \sse U_k\setminus \overline{J^+_{\check{g}_{\eps_k}}(p)}$. Then this curve must intersect $\partial J^+_{\check{g}_{\eps_k}}(p)$ and we choose $q_k$ to be such an intersection point.} As before, passing to a subsequence if necessary, we may assume that the $\gamma_{\eps_k}$ converge to a $g$-geodesic $\gamma :[0,1]\to M$ in $C^1([0,1])$. This limit $\gamma $ must be $g$-null and we have $\tau(p,q)=0$ (else there would exist a maximizing curve from $p$ to $q$ of positive length and since maximizing curves have fixed causal character, see \cite{GL:18}, such a curve would be timelike), so $\tau(p,q)=L(\gamma)$ and $\gamma $ is maximizing.  
\end{proof}

\begin{remark}
The same holds for any net $\hat{g}_\eps$ that converges to $g$ in $C^1_\mathrm{loc}$ and satisfies $g\prec \hat{g}_\eps \prec \hat{g}$ for some globally hyperbolic $\hat{g}$: For $p\ll q$ this is just the first part of the proof. For $p<q$ but not $\ll q$ one only needs to establish the existence of points $q_k \in \partial J^+_{\hat{g}_{\eps_k}}(p)$ converging to $q$. This follows similarly as in the $\check{g}_\eps$ case but needs $\bigcap_{\eps} I^+_{\hat{g}_{\eps}}(p)=J^+_{g}(p)$ which can be shown using the same basic limit geodesic argument as in the proof of the previous Lemma.
\end{remark}

Using essentially the same proof, just replacing $p$ with $\Sigma$ or $S$ we also get the following results about the existence of maximizing geodesics to a Cauchy hypersurface $\Sigma$ or a closed, spacelike $(n-2)$-dimensional submanifold $S$.

\begin{Corollary}\label{cor:sigmamaxgeod} Let $(M,g)$ be globally hyperbolic with a $C^1$-metric $g$ and let $\Sigma\sse M$ be a smooth spacelike Cauchy hypersurface. Then for any $q\in I^+(\Sigma)$ there exists at least one timelike geodesic from $\Sigma $ to $q$ maximizing the distance to $\Sigma$. Further, such a geodesic can be obtained  as a $C^1$-limit of a sequence of $\check{g}_{\eps_n}$-geodesics $\gamma_{\eps_n}$ maximizing the $g_{\eps_n}$-distance to $\Sigma$ and thus must in particular start orthogonally to $\Sigma$. 
\end{Corollary}
\begin{Corollary}\label{cor:Smaxgeod} Let $(M,g)$ be globally hyperbolic with a $C^1$-metric $g$ and let $S \sse M$ be a closed, spacelike $(n-2)$-dimensional submanifold. Then for any $q\in J^+(S)\setminus I^+(S)$ there exists at least one null geodesic from $S$ to $q$ maximizing the distance to $S$. Further, such a geodesic can be obtained  as a $C^1$-limit of a sequence of $\check{g}_{\eps_n}$-null $\check{g}_{\eps_n}$-geodesics $\gamma_{\eps_n}$ maximizing the $g_{\eps_n}$-distance to $S$ and thus must in particular start orthogonally to $S$. 
\end{Corollary}

These are easily recognized as corresponding to the situations in the Hawking theorem and the Penrose theorem, respectively.

\subsection{Stability of causal geodesic completeness}
In this section we will refine the existence time results of the previous section to obtain stability of causal (or null) geodesic completeness under the additional assumption of global hyperbolicity.

This requires us to work with a finer topology than the $C^1_\mathrm{loc}$ one we used in the previous section: 
Given a manifold $M$ we denote by $\mathrm{Pseud}^k(M)$ the space of all semi-Riemannian $C^k$-metrics on $M$. The fine $C^r$-topologies (sometimes also referred to as Whitney topologies)
can be defined on $\mathrm{Pseud}^k(M)$ ($k\geq r$) using the following basis of neighborhoods:  Fix a countable, locally finite atlas $\{U_\alpha,\psi_{\alpha} \}$ such that each
of the charts has compact closure in a larger chart. 
A basis $\{W(g,f)\}$ for the fine $C^r$–topology 
is obtained by taking semi-Riemannian metrics $g$ and positive functions $f\in C^0(M)$ and setting $W(g,f)$ to be the set of all $g'\in \mathrm{Pseud}^k(M)$  such that all components and mixed partials up to order $r$ for $g$ and $g'$ in the fixed given atlas are $f$-close. For the $C^1$-fine topology this just means that
\[||g_{ij}-g'_{ij}||_{\infty,U_\alpha}\leq f \; \text{ and }\; ||\partial_rg_{ij}-\partial_rg'_{ij}||_{\infty,U_\alpha}\leq f. \]
Clearly $C^1$-fine convergence implies convergence in $C^1_\mathrm{loc}$, but the converse is not true.

In \cite{BEStab}, Beem and Ehrlich investigated the stability of geodesic completeness for smooth semi-Riemannian metrics and showed in particular that any null (resp.~causal) geodesically complete globally hyperbolic smooth Lorentzian metric $g$ possesses a $C^1$-fine  neighborhood $U$ in the space $\mathrm{Lor}^\infty(M)$ of smooth Lorentzian metrics such that each $g'\in U$ is null (resp. causal) geodesically complete. 

Using results/techniques from the previous section (Lemmas \ref{lem:aa} and \ref{lem:ab}) to circumvent arguments using uniqueness of geodesics (and thus requiring higher differentiability of the metric) we obtain that any globally hyperbolic $C^1$-Lorentzian metric that is causal geodesically complete possesses a $C^1$-fine  neighborhood $U$ in the space $\mathrm{Pseud}^1(M)$ such that each $g'\in U$ is a causal geodesically complete Lorentzian metric.  While this is a bit of a tangent and will not be required for the proof of the $C^1$-singularity theorems, it provides a neat stability result in a very  natural setting: Contrary to the older result of \cite{BEStab} the regularity of the metrics now coincides with the regularity of the convergence.

\begin{Theorem}\label{thm:stability}
Let $(M,g)$ be globally hyperbolic with a $C^1$-metric $g$ which is causal geodesically complete. Then there exists a $C^1$-fine neighborhood $U\subseteq \mathrm{Pseud}^1(M)$ around $g$ such that all $g'\in U$ are globally hyperbolic Lorentzian metrics that are causal geodesically complete.
\end{Theorem}
\begin{proof}
First note that by \cite[Theorem~4.5]{Clemens} there exists a smooth metric $\hat{g}$ with $g\prec \hat{g}$ and $\hat{g}$ globally hyperbolic. Given any such $\hat{g}$ there always exists a $C^0$-fine (and hence $C^1$-fine) neighborhood $U\subseteq \mathrm{Pseud}^1(M)$ around $g$ such that all $g'\in U$ are Lorentzian metrics and satisfy $g'\prec \hat{g}$ (this essentially follows from \cite[Lemma~1.4]{Clemens}). Thus all $g'\in U$ are globally hyperbolic and for any $K\subseteq M$ one has $J_{g'}^+(K)\cap J_{g'}^-(K)\subseteq J_{\hat{g}}^+(K)\cap J_{\hat{g}}^-(K)$.

Since $J_{\hat{g}}^+(K)\cap J_{\hat{g}}^-(K)$ is compact for any compact $K$, we can construct a compact exhaustion $\{K_i\}_{i=1}^\infty$ of $M$ 
with the following properties: $K_{i-1} \subseteq K_i^\circ$ and $J_{\hat{g}}^+(K_i)\cap J_{\hat{g}}^-(K_i)=K_i$. Note that for $g'\in 
U$ any $g'$-causal $g'$-geodesic is $\hat{g}$-timelike and hence cannot re-enter a $K_i^\circ$ after leaving it.

 Set $c_0=0$ and inductively define constants $0<c_i<\infty$ such that all $g$-causal $g$-geodesics $\gamma $ starting in $K_i$ with $||\dot{\gamma}(0)||_h\leq c_{i-1}+1$ satisfy $||\dot{\gamma}||_h\leq c_i$ while $\gamma $ remains in $K_i$.\footnote{To see that this is possible note that because of Proposition~\ref{prop:relcompactfunnels} it is sufficient to show that there exists $N\in \N$ such that $\gamma|_{[N,\infty)}\sse M\setminus K_i $ for all such $\gamma $. Assume this were not true, then there exist $g$-causal $g$-geodesics $\gamma_n:[0,\infty)\to M$ with $\gamma_n|_{[0,n)}\sse K_i$ and by a limiting process as in Lemma \ref{lem:ab} we obtain an inextendible $g$-causal $g$-geodesic imprisoned in the compact set $K_i$. This contradicts global hyperbolicity.} 
 Further, by Proposition \ref{prop:relcompactfunnels} and skipping over some of the $K_i$ if necessary, we may assume that all that all $g$-causal $g$-geodesics $\gamma $ starting in $K_{i}$ with $||\dot{\gamma}(0)||_h\leq c_{i-1}+1$ remain in $K_{i+1}^\circ$ for all parameters $0\leq s\leq i$.
 
We have the following
\begin{Lemma}
For all $i$ there exists $\eps_i$ such that for all $g'\in U$ with $||g-g'||_{C^1(K_{i+2}\setminus K_{i-1}^\circ )}\leq \eps_i$ all $g'$-causal $g'$-geodesics $\gamma$ contained in $M\setminus K_{i-1}^\circ $ and starting in $K_i \setminus  K_{i-1}^\circ $ with $||\dot{\gamma}(0)||_h\leq c_{i-1}+1$
\begin{itemize}
	\item[(i)] exist and remain in $K_{i+2}^\circ$ for parameters in $[0,i]$ and
	\item[(ii)] for each such geodesic $\gamma $ there exists $s_\gamma $ such that $\gamma|_{[0,s_\gamma]}\subseteq K_{i} \setminus K_{i-1}^\circ $, $\gamma(s_\gamma)\in \partial K_i \subseteq K_{i+1}\setminus K_i^\circ $ and $||\dot{\gamma}(s_\gamma)||_h\leq c_i+1$.
\end{itemize}
\end{Lemma}
\begin{proof}
	Assume not. Then there exists $i\in \mathbb{N}$, metrics $g'_n\in U$ with $||g-g'_n||_{C^1(K_{i+2}\setminus K_{i-1}^\circ )}\leq 1/n$ and $g'_n$-causal $g'_n$-geodesics $\gamma_n:[0,b_n)\to M\setminus K_{i-1}^\circ $ with $\gamma_n(0)\in K_i \setminus K_{i-1}^\circ $, $\gamma_n(0)\to v\in K_{i}$,  and $||\dot{\gamma}_n(0)||_h\leq c_{i-1}+1$ such that either (i) or (ii) does not hold.
	
	Assume (i) doesn't hold. Then there exist $t_n\leq i$ such that $\gamma_n(t_n)\notin K_{i+2}^\circ $, but still $\gamma_n(t_n)\in K_{i+2}$ and hence $\gamma_n|_{[0,t_n]}\subseteq K_{i+2}\setminus K_{i-1}^\circ$. If $||\dot{\gamma}_n|_{[0,t_n]}||_h$ remains bounded, then, by essentially Lemma \ref{lem:aa}, this implies the existence of a $g$-causal $g$-geodesic $\gamma$ with $\gamma(t)\in \partial K_{i+1}$, so $\notin K_{i+1}^\circ$, for some $t< \lim t_n \leq i$, a contradiction. And if $||\dot{\gamma}_n|_{[0,t_n]}||_h$ does not remain bounded, then, by essentially Lemma \ref{lem:ab}, there exists an inextendible $g$-causal $g$-geodesic that does not exist on $[0,\infty)$, also a contradiction.
	
	Assume now that (ii) doesn't hold. Because each $\gamma_n$ must leave every compact subset by global hyperbolicity of $g'_n$ an $s_n$ such that $\gamma_n|_{[0,s_n]}\subseteq K_{i}\setminus K_{i-1}^\circ \subseteq K_{i+2}\setminus K_{i-1}^\circ$ and $\gamma_n(s_n)\in \partial K_i \subseteq K_{i+1}\setminus K_i^\circ $ must exist. Now, similar to when looking at (i), when assuming $||\dot{\gamma}(s_n)||_h> c_i+1$ we get a contradiction: Either to causal completeness of $g$ (if $||\dot{\gamma}_n|_{[0,s_n]}||_h$ does not remain bounded), or to $||\dot{\gamma}||_h\leq c_i$ (in case $||\dot{\gamma}_n|_{[0,s_n]}||_h$ remains bounded and $\dot{\gamma}:[0,s)\to TM$ can be extended to $s=\lim s_n$, as then $||\dot{\gamma}(s)||_h=\lim ||\dot{\gamma}_n(s_n)||_h$), or to the non-existence of inextendible causal curves imprisoned in a compact set (in case $||\dot{\gamma}_n|_{[0,s_n]}||_h$ remains bounded but $\dot{\gamma}:[0,s)\to TM$ cannot be extended to $s=\lim s_n$).
\end{proof}
The desired result now follows almost immediately: Choose a function $f$ such that for all $g'\in W(g,f)$ we have $||g-g'||_{C^1(K_{i+2}\setminus K_{i-1}^\circ)}\leq \eps_i$ for all $i$. Let $g'\in U\cap W(g,f)$ and let $\gamma :(a,b)\to M$ be an inextendible $g'$-causal $g'$-geodesic. We show that $\gamma $ is complete. First, there exists $N\in \mathbb{N}$ such that $\gamma \subseteq M\setminus \bigcup _{i< N} K_i$ and $\gamma$ meets $K_N$. Reparametrizing we can assume $\gamma(0)\in K_N$ and $||\dot{\gamma}(0)||_h\leq c_{N-1}+1$. From the lemma we get that $\gamma $ exists on $[0,N]$, hence $b\geq N$. Further the lemma also tells us that $\gamma|_{[s_\gamma,b)}$ satisfies the assumptions necessary to apply the lemma with $N+1$ instead of $N$, hence $b\geq N+1$. Proceeding like this, we get $b\geq N+n$ for all $n\in \N$, hence $\gamma$ is future complete. Past completeness follows analogously.
\end{proof}

\begin{remark}
	As in \cite{BEStab} the result is also true for null geodesic completeness and it should also be
	sufficient to assume that $g$ is causally disprisoning and causally pseudoconvex (instead of globally hyperbolic), but one would have to be a bit more careful in the construction of the compact exhaustion and one needs to check that, as in \cite[Lemma~3.1]{BEStab}, being causally disprisoning and causally pseudoconvex is $C^1$-fine stable for $g\in C^1$.
\end{remark} 
 
For the remainder of the paper we will return to using $C^1_\mathrm{loc}$ convergence instead of $C^1$-fine convergence because we only ever need estimates on compact sets.

\section{Curvature for $C^1$-metrics} \label{sec:curv}
Since we do not necessarily wish to assume that $g\in W^{2,p}_{\mathrm{loc}}$ for some $p$,\footnote{
If $g\in C^1\cap W^{2,p}_{\mathrm{loc}}$, then curvature is well defined in $L^p_{\mathrm{loc}}$. Any reader only interested in this case can safely skip most of this section and just read all curvature bounds, in particular those in Def.~\ref{def:sec} and Def.~\ref{def:genNEC}, as pointwise a.e.~bounds.} the definition of curvature is not completely straightforward. Nevertheless, $C^1$ is well above the minimal regularities considered in general relativity\footnote{The most general class of metrics that allow for a sensible definition of distributional curvature was identified in \cite{GT87} to be $g\in W^{1,2}_{\mathrm{loc}}$ satisfying an appropriate non-degeneracy condition (see also \cite{SV09}). Any $C^1$ metric belongs to this class.} and it is easy to see that the Riemann tensor and the Ricci tensor as well as scalar curvature are still well defined \emph{as distributions}. 

In this section we are going to very briefly review what exactly this entails, focusing especially on what it means for a distribution to be bounded (by a constant or another distribution), how one may formulate bounds on timelike Ricci curvature equivalent to the strong energy condition in the smooth case\footnote{We postpone looking at possible ways to define the bounds on Ricci curvature necessary to prove Penrose's theorem to section \ref{sec:penrose}.} and how one regularizes both distributions and non-smooth functions/tensor fields on manifolds. Nothing in this section is new, we mainly follow \cite{TheBook_GKOS}, Section 3.1, but some information can also be found in, e.g., \cite{DieudIII, deRham}. See also \cite{LeFlochMadare} for a concise summary on how to define distributional curvature for more general connections in case $M$ is orientable.

\subsection{Very brief review: Distributions on manifolds}
Analogous to the definition of distributions on an open subset of
$\mathbb{R}^{n}$ the space of distributions on a manifold $M$ is
defined as the dual space of certain smooth 'test objects'. However,
if one wants to preserve the embedding of smooth functions into distributions
via integration, the space of these test objects has to be the space
of compactly supported smooth sections of the \emph{volume bundle} $\vol M$ over $M$,
rather than $C^\infty_c\left(M\right)$ itself. 

If $M$ is orientable, then $\mathrm{Vol}\left(M\right)$ is simply the well-known vector bundle of exterior $n$-forms $\Lambda^{n}T^{*}M$. For non-orientable manifolds $\vol M$ is still a one-dimensional real vector bundle over $M$, but it is characterized by the transition functions $
\phi_{\alpha\beta}(x)=\left|\det D\left(\psi_{\beta}\circ\psi_{\alpha}^{-1}\right)\left(\psi_{\alpha}(x)\right)\right|$ 
 (instead of $\phi_{\alpha\beta}(x)=\det D\left(\psi_{\beta}\circ\psi_{\alpha}^{-1}\right)\left(\psi_{\alpha}(x)\right)$ which are the transition functions for $\Lambda^{n}T^{*}M$). Locally, any smooth section of the volume bundle (i.e., a volume density) can be written as \[\mu|_U=f \, |dx^1\wedge \dots \wedge dx^n|\] for a smooth function $f$ on $U$. 
 
  The integral of a compactly supported volume density $\mu\in\Gamma_{c}\left(M,\vol M\right)$, with local expressions $^\alpha\mu |dx^1\wedge \dots \wedge dx^n|$,
is defined analogous to the integration of an $n$-form on an orientable
manifold by
\[
\int_{M}\mu=\sum_{\alpha}\int_{U_{\alpha}}\xi_{\alpha}\mu:=\sum_{\alpha}\int_{\psi_{\alpha}\left(U_{\alpha}\right)}\xi_{\alpha}\left(\psi_{\alpha}^{-1}(x)\right)\,^\alpha\mu \,\left(\psi_{\alpha}^{-1}(x)\right)dx
\]
for a given atlas $\left(U_{\alpha},\psi_{\alpha}\right)$ and a subordinate
partition of unity $\xi_{\alpha}$. This does not depend on the choice of the atlas or the partition of unity. The volume density $\mu$ is called \emph{non-negative} (or \emph{positive}) if every $^\alpha\mu \in C^\infty(U_\alpha)$ is a non-negative (or positive) function. This is equivalent to $\int_U \mu \geq 0$ (resp.~$>0$) for any open $U\subseteq M$.
\\

The space of
\emph{distributions} on $M$ is then
defined as the space of continuous linear functionals on $\Gamma_{c}\left(M,\vol M\right)$,\footnote{Details on the topology on $\Gamma_{c}\left(M,\vol M\right)$ can be found in \cite{TheBook_GKOS}.
However, these details will not be important here.}  i.e.,
\[
\mathcal{D}'\left(M\right):=\Gamma_{c}\left(M,\vol M\right)'.
\]
We can embed $C^\infty \hookrightarrow \mathcal{D}'\left(M\right)$ by identifying a smooth function $f$ with the functional $\mu \mapsto \left\langle f,\mu \right \rangle :=\int_M f\mu $. 
More generally, for the vector bundle $T^r_sM$ over $M$, one defines the space of 
\emph{$T^r_sM$-valued distributions} or \emph{$(r,s)$-tensor distributions} by
\[
\mathcal{D}'\mathcal{T}_{s}^{r}(M)\equiv\mathcal{D}'\left(M,T_{s}^{r}M\right):=\Gamma_{c}\left(M,T_{r}^{s}M\otimes\vol M\right)'.
\]
We will mainly use the following equivalent description of tensor distributions:
\[
\mathcal{D}'\mathcal{T}_{s}^{r}(M) \cong\mathcal{D}'\left(M\right)\otimes_{\mathcal{C}^{\infty}}\mathcal{T}_{s}^{r}\left(M\right).
\]

The next proposition (cf.~\cite[Section 3.1.4]{TheBook_GKOS})
will allow us to do everything in charts and basically treat tensor distributions like families of distributions on open balls in $\R^n$.

\begin{Proposition}\label{prop:distribcharts} \begin{itemize}
\item[(i)] Given an atlas $\left(U_{\alpha},\psi_{\alpha}\right)$ of $M$, locally $\res T{U_\alpha}\in\mathcal{D}'\left(U_{\alpha},T_{s}^{r}M\right)$
can be written as
\[
\res T{U_{\alpha}}=\left(^\alpha T\right)_{j_{1}\dots j_{s}}^{i_{1}\dots i_{r}}\,\partial_{i_{1}}\otimes\dots\otimes\partial_{i_{r}}\otimes dx^{j_{1}}\otimes\dots\otimes dx^{j_{s}}
\]
with local coefficients $\left(^\alpha T \right)_{j_{1}\dots j_{s}}^{i_{1}\dots i_{r}}\in\mathcal{D}'\left(U_{\alpha}\right)$.

\item[(ii)] Extending the pushforward $(\psi_{\alpha})_*:C^\infty_c(U_\alpha)\to C^\infty_c(\psi_\alpha(U_\alpha))$ to distributions (by setting $ \left\langle (\psi_\alpha)_{*}u,\phi \right\rangle :=\left\langle u,(\Psi_\alpha)^{*}\phi \right\rangle $, where $\Psi_\alpha$ is a chart for $\vol M$ associated to $\psi_\alpha$) further allows us to identify $\left(^\alpha T\right)_{j_{1}\dots j_{s}}^{i_{1}\dots i_{r}}\in \mathcal{D}'(U_\alpha)$ with its chart representations \[(^{\alpha}\tilde{T})_{j_{1}\dots j_{s}}^{i_{1}\dots i_{r}}:=(\psi_{\alpha})_*\big(\left(^{\alpha}T\right)_{j_{1}\dots j_{s}}^{i_{1}\dots i_{r}}\big) \in \mathcal{D}'(\psi_\alpha(U_\alpha)).\]

\item[(iii)] As in the case of smooth tensor fields, this process is reversible: For any family $\{\big((^{\alpha}\tilde{T})_{j_{1}\dots j_{s}}^{i_{1}\dots i_{r}}\big) \in \mathcal{D}'(\psi_\alpha(U_\alpha))^{n^{r+s}}\}_{\alpha \in A}$ obeying the right transformation rules there exists a tensor distribution $T$ on $M$ such that 
\[(\psi_{\alpha})_*\big(\left(^{\alpha}T\right)_{j_{1}\dots j_{s}}^{i_{1}\dots i_{r}}\big) =(^{\alpha}\tilde{T})_{j_{1}\dots j_{s}}^{i_{1}\dots i_{r}} \in \mathcal{D}'(\psi_\alpha(U_\alpha)).\]
\end{itemize}
\end{Proposition}

\subsection{Distributional curvature and curvature bounds}
As justified by the exposition in the previous subsection, we may do all our calculations and definitions locally. 
	So let $\Omega $ be an open ball in $\R^n$ and let $g_{ij} \in C^1(\Omega)$ be a Lorentzian metric on $\Omega$. Then the inverse exists and is again $C^1$ and all Christoffel symbols 
	\begin{equation}\label{eq:Gamma}
	\Gamma^k_{ij}= \frac{1}{2} g^{kl}(\partial_ig_{jk}+\partial_jg_{ik}-\partial_kg_{ij})
	\end{equation} are continuous on $\Omega$. Hence, any partial derivatives of the Christoffel symbols are well defined as distributions on $\Omega$ in the usual way (i.e., $\langle \partial_m \Gamma^k_{ij} ,\phi \rangle := -\int_\Omega \Gamma^k_{ij} \partial_m\phi $). We now define the distributional Riemann tensor by the usual expression, i.e.,
\begin{equation}\label{eq:Riem}
\mathrm{Riem}^m_{\;\;\;ijk} := \partial_j \Gamma^m_{ik}-\partial_k \Gamma^m_{ij}+\Gamma^m_{js}\Gamma^s_{ik}-\Gamma^m_{ks}\Gamma^s_{ij}\in \mathcal{D}'(\Omega ).
\end{equation}
If $g_{ij}$ comes from a $C^1$-metric on a manifold $M$, one can check (as in the smooth case) that this expression indeed behaves correctly under a change of charts, so all such local expressions together define a tensor distribution $\mathrm{Riem}[g]$ on $M$. 

To obtain the distributional Ricci tensor we simply contract the distributional Riemann tensor the usual way, i.e., $\Ric_{ij}:=\mathrm{Riem}^m_{\;\;\;imj}$ or in coordinates
\begin{equation}\label{eq:Ric}
\Ric_{ij} := \partial_m \Gamma^m_{ij}-\partial_j \Gamma^m_{im}+\Gamma^m_{ij}\Gamma^k_{km}-\Gamma^m_{ik}\Gamma^k_{jm}\in \mathcal{D}'(\Omega ).
\end{equation}
Note that this is entirely unproblematic because it does not involve any type changes. Again this definition gives rise to a tensor distribution $\Ric[g]$ on $M$. 

If we want to define scalar curvature by $\mathrm{R}=g^{ij}\Ric_{ij}$ we would multiply a distribution with a non-smooth function which is problematic. The way around this is to instead define
\begin{equation*}
\mathrm{R} := \partial_m (g^{ij} \Gamma^m_{ij})-(\partial_m g^{ij}) \Gamma^m_{ij}-\partial_j (g^{ij} \Gamma^m_{im})+(\partial_j g^{ij})\Gamma^m_{im}+g^{ij}\Gamma^m_{ij}\Gamma^k_{km}-g^{ij}\Gamma^m_{ik}\Gamma^k_{jm},
\end{equation*}
which, thanks to the product rule, is an equivalent expression in the smooth case.

\paragraph{Curvature bounds:}
For distributions $u\in \mathcal{D}'(M)$ one has a natural notion of positivity, namely
\begin{definition}\label{def:posdistribs}
	Let $u\in \mathcal{D}'(M)$. Then $u\geq 0$ (resp.~$u >0$) if $\langle u,\mu \rangle \geq 0$ (resp.~$>0$) for any compactly supported non-negative (resp.~positive) volume density $\mu $.
	
	Using this, one defines $u\geq v$ (resp.~$u<v$) for any two $u,v\in \mathcal{D}'(M)$ by non-negativity (rep.~positivity) of the distribution $u-v$.
\end{definition}

This can be used to define energy conditions the following way: Given any smooth vector field $\MX$, we define $\Ric (\MX,\MX)\in \mathcal{D}'(M)$ locally via $\Ric(\MX,\MX)=\Ric_{ij} \MX^i \MX^j$.   

\begin{definition}\label{def:sec}
	Let $g$ be a $C^1$-Lorentzian metric on $M$. We say that $(M,g)$ satisfies the \emph{strong energy condition} or has \emph{non-negative timelike Ricci curvature}, if the distribution $\Ric(\MX,\MX)$ is non-negative (in the natural sense of Def.~\ref{def:posdistribs}) for all smooth timelike vector fields $\MX$ on $M$. 
\end{definition}

\begin{remark}\begin{itemize}
		\item[(i)]	If $g$ is smooth, this is equivalent to the pointwise condition $\Ric_p(X,X)\geq 0$ for all timelike vectors $X\in T_pM$ and all $p\in M$ because any such $X$ can be extended to a smooth timelike vector field $\MX$ on $M$ and the smooth function $\Ric(\MX,\MX)\in \mathcal{D}'(M)$  is non-negative in the sense of Def.~\ref{def:posdistribs} if and only if it is non-negative at every point. If $g\in C^{1,1}$, this is equivalent to the a.e. pointwise condition used in \cite{hawkingc11}.
		\item[(ii)] The strong energy condition is equivalent to all local representations $\Ric_{ij} \MX^i \MX^j\in \mathcal{D}'(\Omega)$ being non-negative distributions on $\Omega $ for any smooth timelike vector field $\MX$ on $\Omega $.
	\end{itemize}
\end{remark}

\subsection{Regularization of distributions and non-smooth tensor fields}
A distribution $u\in \mathcal{D}'(\Omega)$ with compact support\footnote{The support of a distribution  $u\in \mathcal{D}'(\Omega)$ is defined as follows: $x \in \mathrm{supp}(u)$ if and only if for all open neighborhoods $V\sse \Omega$ of $x$ there exists some $\phi \in \mathcal{D}(\Omega)$ with $\mathrm{supp}(\phi)\sse V$ with $\langle u, \phi \rangle\neq 0$.  If $\chi$ is a compactly supported smooth function on $\Omega$, then $\chi u$ has compact support (in fact $\mathrm{supp}(\chi u)\subseteq \mathrm{supp}(\chi)$) for any $u\in \mathcal{D}'(\Omega)$}
	can be \emph{regularized} by convolution as follows: Let $\rho_\eps$ be a standard mollifier, then for $\eps$ small enough (i.e., $\eps <\mathrm{d}(\mathrm{supp}(u),\partial \Omega)$) the function
	\begin{equation}\label{eq:convdistribRn}
	x\mapsto (u\star \rho_\eps)(x):=\langle u, \rho_\eps(x-.) \rangle
	\end{equation}
is smooth on $\Omega$. Further, $\langle u\star \rho_\eps, \phi \rangle \to \langle u,\phi \rangle$ for any test function $\phi$ and if $u\geq 0$, then $u\star \rho_\eps\geq 0$.
If $u=f\in L^{1}(\Omega)$, then this is just the usual convolution, i.e.,
$(f\star \rho_\eps)(x)=\int_{\Omega}f(x-y)\rho_\eps(y) dy$.\\

We now want to lift this to distributions on manifolds:\footnote{See also \cite[Theorem~3.2.10]{TheBook_GKOS} for this construction.}  We fix a countable atlas $\{\left(U_{\alpha},\psi_{\alpha}\right)\}_{\alpha\in \N}$ with relatively compact $U_\alpha$, a subordinate partition of unity $\xi_{\alpha}:M\to [0,1]$ and functions $\chi_{\alpha}\in C^{\infty}_c(U_{\alpha})$
with $\left|\chi_{\alpha}\right|\leq1$ and $\chi_{\alpha}\equiv1$
on an open neighborhood of $\supp(\xi_{\alpha})$ in $U_{\alpha}$. Then for $\T\in\mathcal{D}'\mathcal{T}_{s}^{r}(M)$ 
we define a smooth $(r,s)$-tensor field $\T\star_M \rho_{\eps}$ via
\begin{equation}
\T\star_M \rho_\eps:=\sum_{\alpha\in\N}\chi_{\alpha}(\psi_{\alpha})_*^{-1} \big(\left((\psi_{\alpha})_*(\xi_\alpha \T)\right)*\rho_{\eps}\big), \label{eq:smoothing}
\end{equation}
where $(\psi_{\alpha})_*(\T):=\big((^\alpha\T)_{j_{1}\dots j_{s}}^{i_{1}\dots i_{r}}\big) \in (\mathcal{D}'(\psi_\alpha(U_\alpha)))^{n^{r+s}}$ (see Proposition~\ref{prop:distribcharts})
and $(\psi_{\alpha})_*(\T)\star \rho_{\eps}$ is to be understood component-wise via \eqref{eq:convdistribRn}.

Clearly
for any distribution $u\in \mathcal{D}'(M)$ we have
\begin{equation}\label{eq:posofsmoothing} u \geq 0 \implies u\star_M \rho_\eps \geq 0 .
\end{equation} 

Extending definition \eqref{eq:smoothing} to $C^k$- or $W^{k,p}_\mathrm{loc}$-tensor fields or functions on $M$ one recovers the usual way of smoothing tensor fields of these regularities. We have the following well-known convergence properties (which follow easily from the corresponding properties of smoothing via convolution on $\R^n$):

\begin{Proposition}[Convergence of $\T\star_M \rho_\eps$]\label{prop:convofconvol} We have:
\begin{itemize}
\item[(i)] If $\T\in \mathcal{D}'\mathcal{T}_{s}^{r}(M)$, then
\[\langle \T\star_M \rho_\eps,\phi \rangle \to \langle \T,\phi \rangle \qquad \forall \phi\in\Gamma_{c}\left(M,T_{r}^{s}M\otimes\vol M\right). \]
\item[(ii)] If $\T$ is $C^k$, then $\T\star_M \rho_\eps \to \T$ in $C^{k}_{\mathrm{loc}}$. Further, if $k\geq 1$, then 
for any compact $K\subseteq M$ there exists a constant $c_K>0$ and $\eps_0(K)$ such that
\begin{equation}\label{eq:convollesseps}
||\T-\T\star_M \rho_\eps||_{\infty,K}\leq c_K \eps
\end{equation}
for all $\eps<\eps_0$.

\item[(iii)] If $\T$ is $W^{k,p}_{\mathrm{loc}}$ ($1\leq p<\infty$), then $\T\star_M \rho_\eps \to \T$ in $W^{k,p}_{\mathrm{loc}}$.
\end{itemize}
\end{Proposition}
\begin{proof}
	We only give a brief sketch for \eqref{eq:convollesseps}, the rest follows similarly. It suffices to show that for any compact $K \sse U$ for some chart chart $(\psi,U)$ we have
$||\psi_*\T-\psi_*(\T\star_M \rho_\eps)||_{\infty,K}\leq c_K \eps$. Since $\chi_\alpha\equiv 1$ on a neighborhood of $\supp \xi_\alpha$ and $|\{\alpha:K\cap U_\alpha\neq 0\}|$ is finite, there exists $\eps_0(K)$ such that for all $\eps <\eps_0(K)$
\[\psi_*(\T\star_M \rho_\eps)=\sum_{\{\alpha:K\cap U_\alpha\neq 0\}} \psi_*\big((\psi_{\alpha})_*^{-1} \big(\left((\psi_{\alpha})_*(\xi_\alpha \T)\right)\star \rho_{\eps}\big)\big).\]
Also,
\[\psi_*\T=\psi_*(\sum_{\{\alpha:K\cap U_\alpha\neq 0\}} \xi_{\alpha}\T) =\sum_{\{\alpha:K\cap U_\alpha\neq 0\}} \psi_*((\psi_{\alpha})_*^{-1}(\psi_{\alpha})_*(\xi_\alpha \T))\]
and thus, since $||f-f\star \rho_\eps||_{\infty}\leq ||\partial f||_{\infty} \eps =c_K\eps$ if $f\in C^1(\R^n)$ with compact support, we see that the difference of each summand will be bounded by some positive constant $c_{K,\alpha}$ times $\eps$ and the claim follows.
\end{proof}

\section{Hawking's singularity theorem} \label{sec:hawk}

\subsection{Curvature bounds and approximations: A $C^1$-version of \cite[Lemma~3.2]{hawkingc11}}
One of the crucial observations in proving singularity theorems for $C^{1,1}$-metrics was that, while the non-negative timelike Ricci curvature condition itself does not survive the approximation process, nearby approximating metrics still have to have ''almost non-negative'' timelike Ricci curvature in the following sense (cf. \cite[Lemma~3.2]{hawkingc11}):

\begin{Lemma}\label{lem:c11hawk}
Let $M$ be a smooth manifold with a $C^{1,1}$-Lorentzian metric $g$ and smooth background Riemannian metric $h_M$ on $M$. Let $K$ be a compact subset of $M$ and suppose that $\Ric(\MX,\MX)\geq 0$  for every $g$-timelike smooth local vector field $\MX$. Then
\begin{equation}\label{suffest}
\begin{split}
  &\forall C>0 \; \forall \delta >0 \;  \forall \kappa<0\ \exists \eps_0>0\ \forall \eps<\eps_0\
  \forall X\in TM|_K \\ 
  & \text{ with }\ g(X,X)\le \kappa 
   \text{ and } \|X\|_{h} \leq C:
  \ \Ric_\eps(X,X) > -\delta.
  \end{split}
\end{equation}
Here 
$\Ric_\eps$ is the Ricci-tensor corresponding to a smooth approximating metric
as in \cite[Proposition~3.1]{hawkingc11}.
\end{Lemma}

Before we can establish our own version of this for $C^1$-metrics, we will need some preparations: \\

To begin with, we state and prove our own version of \cite[Proposition~2.5]{KSSV} (which itself goes back to, essentially, \cite[Proposition~1.2]{CG})
adapted to the precise regularity of our metric and the precise convergence properties we will need. We will postpone a discussion comparing \cite[Proposition~2.5]{KSSV} and Lemma \ref{lem:approxnice} to Remark~\ref{rem:compapprox}.

\begin{Lemma}\label{lem:approxnice}
	Let $(M,g)$ be a spacetime with a $C^1$-Lorentzian metric.
	Then for any $\eps>0$, there exist smooth Lorentzian metrics $\check{g}_\eps$ and $\hat{g}_\eps$, time-orientable by the same smooth timelike vector field as $g$, on $M$ such that \begin{itemize}
		\item[(i)] 	$\check{g}_\eps \prec g\prec \hat{g}_\eps$ for all $\eps$
		\item[(ii)] $\check{g}_\eps - g\star_M \rho_{\eps}\to 0$ in $C^{\infty}_{\mathrm{loc}}$ and for any compact $K\sse M$ there exists $c_K>0$ and $\eps_0(K)>0$ such that
		\begin{equation*}
		||\check{g}_\eps - g\star_M \rho_{\eps}||_{\infty,K}\leq c_K \eps
		\end{equation*} for all $\eps <\eps_0(K)$. The same also holds for $\hat{g}_\eps$.
		\item[(iii)] $\check{g}_\eps$ and $\hat{g}_\eps$ converge to $g$ in $C^1_\mathrm{loc}$ 
	\end{itemize}
\end{Lemma}
\begin{proof}
	This essentially follows by doing the same construction as in the proof of \cite[Proposition~2.5]{KSSV}. In the spirit of making the present paper more self-contained, we nevertheless include the full construction of $\check{g}_\eps$ here (especially since we are going to use $\check{g}_\eps$ while the proof in \cite{KSSV} only explicitly constructs $\hat{g}_\eps$). Also, to get (ii) one has to tweak their arguments a bit: Essentially, the idea is to shift $\eta_\alpha(\eps)\equiv \eta(\lambda_\alpha(\eps),\alpha)$ from the ''convolution term'' to the ''correction term'' in the construction.

	We use the same notation as in section \ref{sec:curv} 
but will additionally assume that each chart $\psi_\alpha$ can be extended to a chart $\psi_\alpha'$ on an open set $U'_\alpha \supset \overline{U_\alpha}$ and $\overline{U_\alpha}$ is compact.
	
	Let $\omega $ be a smooth $g$-timelike one-form on $M$ (exists by time-orientability, see \cite[Proposition~2.5]{KSSV} for a more detailed argument). 
	Fix $\alpha$ and let $\{^\alpha\nu_i\}_{i=1,\dots,n-1}$ be smooth $g$-spacelike one-forms on $\psi_\alpha'(U'_\alpha)$ such that $\{(\psi_\alpha')_*(\omega)|_p,\,^\alpha\nu_1|_p,\dots,\,^\alpha\nu_{n-1}|_p\}$ forms a basis of $\R^n$ for all $p\in \psi_\alpha'(U'_\alpha)$. On $V_\alpha :=\psi_\alpha(U_\alpha)\sse \R^n$ set $^\alpha v :=\sum_{i=1}^{n-1} \,^\alpha\nu_i \otimes \,^\alpha\nu_i $ and $^\alpha g_{\eps}:=((\psi_\alpha)_*(\xi_\alpha g))\star \rho_{\eps}$ and define
	\begin{equation}\label{eq:defcheckgeps}
	^\alpha\tilde{g}_\eps :=\,^\alpha g_{\eps}+\eta_\alpha(\eps ) \,^\alpha v.
	\end{equation}
We want to determine $\eta_\alpha(\eps)$ such that $^\alpha\tilde{g}_\eps \prec (\psi_\alpha)_*g$ on $V_\alpha$. This follows immediately if we can show that $((\psi_\alpha)_*g)|_p(X,X)\geq 0$ implies $^\alpha\tilde{g}_\eps|_p (X,X)>0$ for all $p\in V_\alpha$ and $X\in \R^n$ with $||X||_e=1$. So, let $p$ and $X$ be as above, then $^\alpha v|_p(X,X)=\sum (^\alpha\nu_i|_p(X))^2>0$ (else $X$ would have to be proportional to the vector field dual to $\omega $, hence $g$-timelike, contradicting $((\psi_\alpha)_*g)|_p(X,X)\geq 0$). By (relative) compactness there exists $\kappa_\alpha >0$ such that $^\alpha v|_p(X,X)>\kappa_\alpha $ for all  $p\in V_\alpha$ and $X\in \R^n$ with $||X||_e=1$. Now combining this with \eqref{eq:convollesseps} 
shows
\begin{align*}
^\alpha\tilde{g}_\eps|_p (X,X)&=(\psi_\alpha)_*(\xi_\alpha g)|_p(X,X)+\big(\,^\alpha g_{\eps}|_p-(\psi_\alpha)_*(\xi_\alpha g)|_p\big) (X,X)+\eta_\alpha(\eps)\,^\alpha v|_p(X,X) \\
&\geq 0+ (-c_{\alpha} \eps)+\eta_\alpha(\eps) \kappa_\alpha>0
\end{align*}
if $\eta_\alpha(\eps)=C_\alpha \eps$ for any constant $C_\alpha >\frac{c_{\alpha}}{\kappa_\alpha}$.

We define $\tilde{g}_\eps :=\sum_{\alpha\in\N} \chi_\alpha (\psi_\alpha)_*^{-1}(^\alpha\tilde{g}_\eps) \in \mathcal{T}^0_2M$. This satisfies
\begin{itemize}
\item[(a)] For all compact $K\sse M$ there exists a constant $c_K>0$ such that $||\tilde{g}_\eps  - g\star_M \rho_{\eps}||_{\infty,K}\leq c_K \eps$: We have
\begin{align*}\tilde{g}_\eps &=\sum_{\alpha\in\N} \chi_\alpha (\psi_\alpha)_*^{-1}(^\alpha\tilde{g}_\eps)=\sum_{\alpha\in\N} \chi_\alpha (\psi_\alpha)_*^{-1}(^\alpha g_{\eps})+\eps \sum_{\alpha\in\N}C_\alpha  \chi_\alpha (\psi_\alpha)_*^{-1}(^\alpha v)\\
&=g \star_M \rho_{\eps}+\eps \sum_{\alpha\in\N} C_\alpha \chi_\alpha (\psi_\alpha)_*^{-1}(^\alpha v),
\end{align*} 
hence $||\tilde{g}_\eps  - g\star_M \rho_{\eps}||_{\infty,K}=\eps ||\sum_{\alpha\in\N}C_\alpha \chi_\alpha (\psi_\alpha)_*^{-1}(^\alpha v)||_{\infty,K}\leq c_K \eps $.
\item[(b)] $\tilde{g}_\eps  - g\star_M \rho_{\eps}\to 0$ in $C^\infty_\mathrm{loc}$ 
\item[(c)] $\tilde{g}_\eps\to g$ in $C^1_\mathrm{loc}$ 
\item[(d)] If $g(X,X)\geq 0$ for some $0\neq X\in TM$, then $\tilde{g}_\eps (X,X)>0$: $g(X,X)\geq 0$ implies that $((\psi_\alpha)_*g)(X_\alpha,X_\alpha)\geq 0$ (where $X_\alpha:=(\psi_\alpha)_*(X)$) for all $\alpha$ with $p\in U_\alpha$, hence $^\alpha \tilde{g}_\eps (X_\alpha,X_\alpha)>0$ and thus $(\psi_\alpha)_*^{-1}(^\alpha \tilde{g}_\eps)(X,X)>0$.
\end{itemize}
Further, by construction, $(\eps,p)\mapsto \tilde{g}_\eps(p)$ is smooth and by (c) and (d) for any compact $K\sse M$ there exists some $\eps_K$ such that for all $0<\eps<\eps_K$, $\tilde{g}_\eps$ is of the same signature as $g$, satisfies $\tilde{g}_\eps|_K \prec g|_K$ and is time-orientable by the same smooth $g$-timelike vector field as $g$. 
Thus (cf.~e.g.~\cite[ Lemma~4.3]{HKS:12}) we can construct a smooth map $u:(\eps,p)\mapsto \R_{>0}$ such
that for each $\eps$, the map $\check{g}_\eps : M\to T^0_2M$ defined by $\check{g}_\eps (p):= \tilde{g}_{u(\eps,p)}(p)$ is a globally defined Lorentzian metric, time-orientable by the same vector field as $g$ and satisfying $\check{g}_\eps \prec g$ and which on any given compact $K\sse M$
coincides with $\tilde{g}_\eps$ for sufficiently small $\eps$. This $\check{g}_\eps$ satisfies (i)-(iii) as desired.
\end{proof}
Similar estimates to \eqref{eq:convollesseps} and Lemma~\ref{lem:approxnice} (ii) hold for the inverse metrics:
\begin{Corollary}
	For any compact $K\sse M$ there exists $c_K>0$ and $\eps_0(K)>0$ such that
	\begin{equation}\label{eq:invest1}
	||(\check{g}_\eps)^{-1} - (g\star_M \rho_{\eps})^{-1}||_{\infty,K}\leq c_K \eps
	\end{equation} and 
		\begin{equation}\label{eq:invest2}
	||g^{-1} - (g\star_M \rho_{\eps})^{-1}||_{\infty,K}\leq c_K \eps
	\end{equation} for all $\eps <\eps_0(K)$. 
\end{Corollary}
\begin{proof} Since both $\check{g}_\eps$ and $g\star_M \rho_{\eps}$ converge locally uniformly to $g$, both nets are locally uniformly bounded and locally uniformly non-degenerate. Using this and \eqref{eq:convollesseps}, respectively Lemma~\ref{lem:approxnice} (ii), the results follow straightforwardly from the general formula $A^{-1}=\frac{1}{\det A} C(A)^{T}$, where $C(A)$ is the matrix of cofactors for $A$.
\end{proof}

\begin{remark}[Comparison with the approximations from \cite{KSSV}]\label{rem:compapprox}
For the estimates in the upcoming Lemma \ref{lem:convolinsteadofcheck} and Lemma \ref{lem:thecurvatureapprox} to work, it will be essential that $\check{g}_\eps-g\star_M \rho_\eps\to 0$ in $C^2_\mathrm{loc}$ and that both \eqref{eq:invest1} and \eqref{eq:invest2} hold. Note that while \eqref{eq:invest2} is independent of the choice of $\check{g}_\eps$, this is not the case for the other two requirements. 

For the approximations constructed in \cite{KSSV} (and used in the proofs of the $C^{1,1}$-singularity theorems), one has $||\check{g}_\eps-g||_{\infty,M}=d_h(\check{g}_\eps ,g)<\frac{\eps}{2}$, which together with \eqref{eq:invest2} implies \eqref{eq:invest1}.
Unfortunately, it is unclear how to estimate the second derivatives of $\check{g}_\eps-g\star_M \rho_\eps$. Defining \[g\star_M\rho_{\eta(\eps)}:= \sum_{\alpha\in\N}\chi_{\alpha}(\psi_{\alpha})_*^{-1} \big(\left((\psi_{\alpha})_*(\xi_\alpha \T)\right)*\rho_{\eta(\lambda_\alpha (\eps),\alpha)}\big)\] their construction gives
$\check{g}_\eps-  g\star_M\rho_{\eta(\eps)}\to 0$ in $C^\infty_\mathrm{loc}$  and \eqref{eq:invest1}. While \eqref{eq:invest2} might not hold for $g\star_M \rho_{\eta(\eps)}$, convergence of $\eta(\lambda_\alpha (\eps),\alpha)\to 0$ as $\eps \to 0$ gives at least $g-g\star_M \rho_{\eta(\eps)}\to 0$ locally uniformly. This is sufficient for the $C^{1,1}$-metrics considered in \cite{hawkingc11} and for the other $C^{1,1}$-singularity theorems
 because if $f$ in Lemma~\ref{lem:Fthesecond} is locally Lipschitz the assertion merely requires $a_\eps \to a$ locally uniformly.

	It is only due to the lower regularity of $g$ we are considering that it becomes essential to have a finer control on how fast $(g\star_M \rho_{\eta(\eps)})^{-1}$ converges to $g^{-1}$. One possibility to obtain fine enough control is following the construction in \cite{KSSV} closely and noting that one can always choose $\eta(\lambda_\alpha (\eps),\alpha)=c_\alpha \eps$ for constants $c_\alpha >0$. This then implies \eqref{eq:invest2} for $g\star_M\rho_{\eta(\eps)}$. However, the path chosen in our construction of $\check{g}_\eps$ in Lemma~\ref{lem:approxnice} allows us to simply use $g\star_M \rho_\eps$ instead of $g\star_M \rho_{\eta(\eps)}$, which improves the readability of follow-up arguments.
\end{remark}

As further preparation for the generalization of Lemma \ref{lem:c11hawk} from $C^{1,1}$-metrics to $C^1$-metrics we show that $\Ric[\check{g}_\eps]-\Ric[g \star_M \rho_{\eps}]\to 0$ locally uniformly. Note that due to the lower regularity this is more complicated than in \cite{hawkingc11} and thus we include it as a separate Lemma (instead of only being one step in the proof of \cite[Lemma~3.2]{hawkingc11}).  

\begin{Lemma}\label{lem:convolinsteadofcheck} $\Ric[\check{g}_\eps]-\Ric[g \star_M \rho_{\eps}]\to 0$ locally uniformly.
\end{Lemma}
\begin{proof} It suffices to check local uniform convergence in any chart $\psi$. Since both $g \star_M \rho_{\eps}$ and $\check{g}_\eps $ converge to $g$ in $C^1_{\mathrm{loc}}$ we have for the products of their Christoffel symbols 
\[\Gamma_{js}^m[\check{g}_\eps ]\Gamma_{ik}^l[\check{g}_\eps ]-\Gamma_{js}^m[g \star_M \rho_{\eps}]\Gamma_{ik}^l[g \star_M \rho_{\eps}]\to 0\]
locally uniformly. To see that $\partial_m \Gamma^k_{ij}[\check{g}_\eps ]-\partial_m\Gamma_{ij}^k[g \star_M \rho_{\eps}]\to 0$ locally uniformly  it suffices to check that $ (\check{g}_\eps)^{mk}  \partial_s\partial_r(\check{g}_\eps )_{ij} - (g \star_M \rho_{\eps}) ^{mk}  \partial_s\partial_r(g \star_M \rho_{\eps})_{ij} \to 0$ locally uniformly for any chart $\psi$ (where we use the shorthand $g_{ij}$ to refer to $(\psi_*g)_{ij}$). We have:
\begin{align*}
&|(\check{g}_\eps)^{mk}  \partial_s\partial_r(\check{g}_\eps )_{ij} - (g \star_M \rho_{\eps}) ^{mk}  \partial_s\partial_r(g \star_M \rho_{\eps})_{ij} |\leq \\
&\leq \underbrace{|(\check{g}_\eps)^{mk}|  |\partial_s\partial_r(\check{g}_\eps -g \star_M \rho_{\eps})_{ij}|}_{\to 0\;\mathrm{by~Lemma~\ref{lem:approxnice}}} + \underbrace{|(\check{g}_\eps)^{mk}-(g \star_M \rho_{\eps})^{mk}|}_{\leq c_K \eps  \;\mathrm{by~\eqref{eq:invest1}}}\; |\partial_s\partial_r(g \star_M \rho_{\eps})_{ij} |.
\end{align*}
So it remains to estimate $\eps |\partial_s\partial_r(g \star_M \rho_{\eps})_{ij} |$. 

In general, on any compact subset $K$ of $U$, we have
\[|\partial_s\partial_r(g \star_M \rho_{\eps})_{ij}|\leq \sum_{\beta} C^{klmd}_{srij}(\psi_\beta,\psi,K)\; \partial_m\partial_d(((\psi_{\beta})_*(\xi_\beta g))_{kl}\star \rho_{\eps})\circ(\psi_\beta\circ\psi^{-1})+C(\psi_\beta,\psi,K), \]
because the transformation terms that come from $\partial_r\partial_sg_{ij} $ not being tensorial will involve at most first 
derivatives of $g \star_M \rho_{\eps}$ which converge locally uniformly and hence are bounded.

So, $\eps ||\partial_s\partial_r(g \star_M \rho_{\eps})_{ij} ||_{\infty,K}$ is bounded by a finite sum of terms that are themselves either bounded by $\eps C\to 0$ or of the form $\eps ||\partial_m(f\star\rho_{\eps})||_{K,\infty}$ for a continuous function $f$ on $\R^n$, which converges to zero as $\eps\to 0$ by Lemma~\ref{lem:Ftheveryfirst}. 
\end{proof}

We are now finally ready to state and prove our own version of the essential Lemma \ref{lem:c11hawk} (\cite[Lemma~3.2]{hawkingc11}) for $C^1$-metrics. Note that there is a slight difference in the formulation to make it more immediately applicable in the situations we will encounter, but both formulations (Lemma~\ref{lem:c11hawk}, only with $g\in C^1$ instead of $g\in C^{1,1}$, and Lemma~\ref{lem:thecurvatureapprox}) are equivalent.

\begin{Lemma}\label{lem:thecurvatureapprox}
Let $M$ be a smooth manifold with a $C^1$-Lorentzian metric $g$. Let $K$ be a compact subset of $TM$ and suppose that the distribution $\Ric(\MX,\MX)\geq 0$  for every $g$-timelike smooth vector field $\MX$ (cf.~Def.~\ref{def:sec}). 
Then
\begin{equation}\label{mine}
\forall \delta >0 \; \exists \eps_0>0\ \forall \eps<\eps_0\
  \forall X \in K  \text{ with }\ \check{g}_\eps (X,X)=-1 :
  \ \Ric[\check{g}_\eps](X,X) > -\delta
\end{equation}
for smooth approximating metrics $\check{g}_\eps$
as in Lemma~\ref{lem:approxnice}.
\end{Lemma}

\begin{proof} Because $\check{g}_\eps \to g$ uniformly on compact sets it is sufficient to establish that \eqref{suffest} remains true for $C^1$-metrics, as this is easily seen to imply \eqref{mine}. To this end, we follow the proof in \cite[Lemma~3.2]{hawkingc11}.

As in their case, it suffices if we prove the result for approximations $g_\eps := g \star_M \rho_{\eps} $,
because $\Ric[\check{g}_\eps]-\Ric[g \star_M \rho_{\eps}]\to 0$ locally uniformly by Lemma~\ref{lem:convolinsteadofcheck}.
As in \cite[Lemma~3.2]{hawkingc11} we proceed by showing that $\Ric[g_\eps ]-\Ric \star_M \rho_{\eps} \to 0$ locally uniformly: 

This is clearly ok for all terms involving only the Christoffel symbols themselves and not their derivatives (i.e., terms involving at most first derivatives of the metric). To deal with the terms involving derivatives of the Christoffel symbols we proceed as in Lemma~\ref{lem:convolinsteadofcheck}: We look at the problem in a chart $\psi$ and careful analysis of the behavior of the relevant terms under a change of charts $\psi_\beta \circ \psi^{-1}$ shows that it is sufficient to show that 
\[((\psi_\beta)_*g_\eps)^{ij} ((\xi \partial_k ((\psi_\beta)_*g)_{lm})\star\rho_{\eps} )-(((\psi_\beta)_*g)^{ij} \xi \partial_k ((\psi_\beta)_*g)_{lm})\star\rho_{\eps} \to 0\] 
in $C^1(\psi_\beta(U_\beta))$ for any smooth function $\xi$ on $\psi_\beta(U_\beta)\sse \R^n$ with compact support in $\psi_\beta(U_\beta)$. 
By \eqref{eq:invest2}, for all compact $K$ there exist $c_K$ such that
$||g^{ij}-g^{ij}_\eps||_{\infty,K}\leq c_K \eps $, so 
this follows from a general version of Friedrichs Lemma (which will be proved in Lemma \ref{lem:Fthesecond}). 

Thus, we may replace $\Ric [g_\eps]$ with $\Ric \star_M \rho_\eps$. Since we are in a compact set the latter is, essentially, just a finite sum of component wise convolutions in charts. So it suffices to look at the case where $M=\R^n$ and the rest of the argument proceeds exactly as in \cite[Lemma~3.2]{hawkingc11}, with only one small modification due to Ric only being distributional: In \cite[eq.~(8)]{hawkingc11} we replace setting $\tilde{\Ric}$ equal to zero outside of a small ball by multiplication with an appropriate positive smooth cut-off function. The rest goes through word for word (keeping in mind \eqref{eq:posofsmoothing}).
\end{proof}

\paragraph{A Friedrichs Lemma and related convergence results.}
We are now going to show the exact version of the Friedrichs Lemma used in the proof of Lemma~\ref{lem:approxnice}. The arguments are based on similar estimates as in \cite[Section 2]{CdLS:12}. We start by showing that if $f$ is only continuous but not Lipschitz, then $\partial_j(f\star \rho_\eps)$ generally diverges as $\eps\to 0$, but it does not diverge too fast:
\begin{Lemma}\label{lem:Ftheveryfirst}
Let $f$ be a continuous function on $\R^n$ and  let $\rho_\eps $ be a standard mollifier. Then $\eps ||\partial_j(f\star\rho_{\eps})||_{K,\infty}\to 0$ as $\eps\to 0$ for any compact set $K\sse \R^n$.
\end{Lemma}
\begin{proof} Since $c\star \partial_j\rho_\eps\equiv 0$ for any constant $c$ we have \[\partial_j(f\star\rho_{\eps})(x)=(f\star \partial_j\rho_{\eps})(x)=((f-f(x))\star \partial_j\rho_{\eps})(x).\]  Let $\tilde{\eps}>0$, then there exists $\eps_0$ such that $|f(x-y)-f(x)|\leq \tilde{\eps}$ for all $x \in K$ and $y$ in the ball of radius $\epsilon_0$ around $0$, i.e., $y\in B_{\eps_0} $. Let $\eps <\eps_0$, then
\begin{align*}
	 \eps |((f-f(x))\star \partial_j \rho_{\eps})(x)|&\leq \eps \int_{B_{\eps}} |f(y-x)-f(x)|\, |\partial_j\rho_{\eps}(y)| dy 
	\leq \eps \tilde{\eps} \int_{B_{\eps}}|\partial_j\rho_{\eps}(y)|dy.
	\end{align*}
	So it is enough to estimate that $\eps\int_{B_{\eps}}|\partial_j\rho_{\eps}(y)|dy$ remains bounded. We have \[\eps \partial_j\rho_{\eps}(y)=\eps \frac{1}{\eps^n}\partial_j(\rho(\frac{y}{\eps}))=\eps\frac{1}{\eps^{n+1}}(\partial_j\rho)(\frac{y}{\eps})=\frac{1}{\eps^{n}}(\partial_j\rho)(\frac{y}{\eps})\]
and thus by the integral transformation formula $\int_{B_{\eps}}|\eps\partial_j\rho_{\eps}(y)|dy=\int_{B_{1}}|(\partial_j\rho)(y)|dy$ is independent of $\eps$.
\end{proof}

Next, 
we show convergence of $(a\star \rho_\eps) (f\star \rho_\eps)- (af)\star \rho_\eps \to 0$ in $C^1_{\mathrm{loc}}$ if $a\in C^1$ and $f\in C^0$ using essentially the same techniques as in the previous proof.
\begin{Lemma}[A Friedrichs Lemma] \label{lem:Fthefirst} Let $a \in C^1(\R^n)$, $f\in C^0(\R^n)$ and $\rho_\eps $ be a standard mollifier. Then $(a\star \rho_\eps) (f\star \rho_\eps)- (af)\star \rho_\eps \to 0$ in $C^1(K)$ for any compact $K\subseteq \R^n$.
\end{Lemma}
\begin{proof}
We only need to show that the first derivatives converge to zero locally uniformly, i.e. we need to estimate
\begin{align*}
h_\eps (x):= \big((af)\star \partial_j\rho_\eps-(a\star \partial_j \rho_\eps) (f\star \rho_\eps)-(a\star  \rho_\eps) (f\star \partial_j\rho_\eps) \big)(x)
\end{align*}
for $j=1,\dots,n$ and $x$ in $K$. We have
\begin{align*}
h_\eps &= ((a-a(x))f)\star \partial_j\rho_\eps+a(x)(f\star \partial_j\rho_\eps)-(a\star \partial_j \rho_\eps) (f\star \rho_\eps)-(a\star  \rho_\eps) (f\star \partial_j\rho_\eps) \\
&= ((a-a(x))f)\star \partial_j\rho_\eps-(a\star \partial_j \rho_\eps) (f\star \rho_\eps)-((a-a(x))\star  \rho_\eps) (f\star \partial_j\rho_\eps) \\
&=  ((a-a(x))(f-f(x)))\star \partial_j\rho_\eps-(a(x)f(x))\star \partial_j\rho_\eps\\
&\,\,\,\,\,-(a\star \partial_j \rho_\eps) ((f-f(x))\star \rho_\eps)
-((a-a(x))\star  \rho_\eps) (f\star \partial_j\rho_\eps) \\
&= ((a-a(x))(f-f(x)))\star \partial_j\rho_\eps-(a\star \partial_j \rho_\eps) ((f-f(x))\star \rho_\eps)
\\&-((a-a(x))\star  \rho_\eps) (f\star \partial_j\rho_\eps),
\end{align*}
where we used $c\star \rho_\eps\equiv c$ and $c\star \partial_j\rho_\eps\equiv 0$ for any constant $c$. Let $C_K:=||\partial a||_{\infty,K+B_1}$, $\tilde{\eps}>0$ and assume that $\eps $ is small enough such that $|f(x-y)-f(x)|\leq \tilde{\eps}$ for all $x \in K$ and $y\in B_\eps $. We now estimate $h_\eps (x)$ term by term: First,
\begin{align*}
& \big|\int ((a(x-y)-a(x))(f(x-y)-f(x)) \partial_j\rho_\eps(y) dy\big|\leq C_K \int |f(x-y)-f(x)| |y| |\partial_j\rho_\eps(y)| dy \\
& \leq C_K \tilde{\eps} \int |y| |\partial_j\rho_\eps(y)| dy \leq  C_K \tilde{\eps} \int \eps  |\partial_j\rho_\eps(y)| dy = C_K' \, \tilde{\eps }.
\end{align*}
Second,
\begin{align*}
& \big|\int a(x-y) \partial_j \rho_\eps(y) dy \int (f(x-y)-f(x))\rho_\eps(y) dy \big|\leq \tilde{\eps } \, \big| \int (a(x-y)-a(x)) \partial_j \rho_\eps (y)dy  \big| \\
&\leq \tilde{\eps } C_K \int |y| |\partial_j \rho_\eps(y) | dy \leq  \tilde{\eps } C_K'. 
\end{align*}
And third,
\begin{align*}
& \big|\int (a(x-y)-a(x)) \rho_\eps(y) dy \int (f(x-y)-f(x)) \partial_j\rho_\eps(y) dy \big|\leq C_K \eps \, \tilde{\eps } \, \int |\partial_j \rho_\eps (y)| dy  \leq  \tilde{\eps } C_K'. 
\end{align*}
\end{proof}

Combining the previous two Lemmas immediately also gives the following:

\begin{Lemma} \label{lem:Fthesecond} Let $a,a_\eps \in C^1_{\mathrm{loc}}(\R^n)$ and assume that $a_\eps \to a$ in $C^1_{\mathrm{loc}}$ and that for all compact $K$ there exists $c_K$ such that $||a_\eps -a||_{\infty, K}\leq c_K \, \eps$. Further, let $f\in C^0(\R^n)$. Then $a_\eps (f\star \rho_\eps)- (af)\star \rho_\eps \to 0$ in $C^1(K)$ for all compact $K\subseteq \R^n$.
\end{Lemma}
\begin{proof} We estimate $a_\eps (f\star \rho_\eps)- (af)\star \rho_\eps = (a_\eps-a\star \rho_\eps) (f\star \rho_\eps)+(a\star \rho_\eps) (f\star \rho_\eps)- (af)\star \rho_\eps $. The first term converges as desired by $C^1$ convergence of $a_\eps -a\star \rho_\eps \to 0$ and Lemma \ref{lem:Ftheveryfirst} 
	(note that $||a_\eps-a\star \rho_\eps||_{\infty,K}\leq c'_K\eps$). 
The second term converges by Lemma \ref{lem:Fthefirst}. 
	\end{proof}

\subsection{Proof of Hawking's theorems for $C^1$-metrics}

Now let $\Sigma\sse M$ be a spacelike hypersurface and $\mathbf{n}$ be the future pointing unit normal vector field along $\Sigma$. Since $g\in C^1$, also $\mathbf{n} \in C^1$. We define the second fundamental form by 
$\mathrm{II}(V,W):=\mbox{nor}(\nabla_VW)$ for $C^1$-vector fields $V,W$  tangential to $\Sigma$ and the mean curvature vector $H$ by \[H:=\frac{1}{n-1} \mathrm{tr}_\Sigma \mathrm{II}=\frac{1}{n-1} \sum_{i=1}^{n-1} \mathrm{II}(E_i,E_i)\]
for (local) orthonormal vector fields $E_i$ tangential to $\Sigma$.

The following lemma is a well known generalization of estimates used in the classical proof of Hawking's theorem. Still, in the interest of completeness we briefly sketch this routine argument here.
\begin{Lemma} \label{lem:hawkingprep} Let $g$ be smooth. Let $\Sigma $ be a spacelike hypersurface and $\gamma :[0,b]\to M$, $p:=\gamma(0)\in \Sigma $ be a future directed unit speed timelike~geodesic that maximizes the distance to $\Sigma $, i.e., $L(\gamma|_{[0,s]})=\tau_\Sigma(\gamma(s))$ for all $s\in[0, b]$. Then $\dot{\gamma}(0)\perp T\Sigma $. Further, if $(n-1) g(H,\dot{\gamma}(0))\geq \beta >0$ and $\Ric(\dot{\gamma},\dot{\gamma})\geq -\delta $ with $0\leq \delta \leq \frac{3\beta}{b}(1-c)$ for some $0<c\leq 1$, then $b\leq \frac{n-1}{c \beta}$.
\end{Lemma}
\begin{proof}
That $\dot{\gamma}(0)\perp T\Sigma $ is standard and does not involve curvature (cf.~e.g.,~\cite[Corollary~10.26]{ON83}). In particular, since $\gamma $ was assumed to have unit speed, $\dot{\gamma}(0)$ is the future pointing unit normal for $\Sigma$ at $p$. The second claim is standard for $\Ric(\dot{\gamma},\dot{\gamma})\geq 0$ (cf.~e.g., \cite[Proposition~10.37]{ON83}), but similar estimates work if $\Ric(\dot{\gamma},\dot{\gamma})\geq -\delta $: If $\gamma $ maximizes the distance to $\Sigma $, it is a maximizing geodesic, hence for any t.l.~variation $\gamma_\eps$ of $\gamma$ with $\gamma_\eps(0)\in \Sigma$ and $\gamma_\eps(b)=\gamma(b)$  we have $\frac{d^2}{d\eps ^2}\big|_{\eps =0}L(\gamma_\eps )\leq  0$. Let $\{e_i\}$ be an ONB of $T_p\Sigma $ and let $E_i$ be the vector field along $\gamma$ obtained by parallel transport of $e_i$. Then, for any variation with variational vector field $V(s):=(1-\frac{s}{b})E_i(s)$ this yields
\[ 0\leq \frac{1}{b} - \int_0^b (1 - \frac{s}{b})^2 g(R(\dot{\gamma}, E_i)E_i, \dot{\gamma}) ds - g(\mathrm{II}(e_i,e_i), \dot{\gamma}(0)). 
\]
Taking the trace gives
\[ 0\leq \frac{n-1}{b} - \int_0^b (1-\frac{s}{b})^2 \Ric(\dot{\gamma},\dot{\gamma}) ds -\beta \leq \frac{n-1}{b}+\delta \frac{b}{3} -\beta \leq \frac{n-1}{b}+\beta (1-c)-\beta=\frac{n-1}{b}-c\beta. \qedhere
\]
\end{proof}
\medskip 
We can now state and prove our first version of Hawking's theorem for $C^1$-metrics. This is a direct generalization of \cite[Theorem~14.55B]{ON83}.
\medskip
\begin{Theorem}[$C^1$-Hawking, version 1]\label{thm:c1hawk1} Let $(M,g)$ be a time-oriented Lorentzian manifold with $g\in C^1$. If $(M,g)$ has non-negative timelike Ricci curvature (in the distributional sense of Def.~\ref{def:sec}) and there exists a smooth compact spacelike hypersurface $\Sigma$ with $g(H,\mathbf{n})>0$
	on $\Sigma$, then M is not timelike geodesically complete (i.e., there exists at least one incomplete timelike geodesic). 
\end{Theorem}
\begin{proof}
Assume to the contrary that $(M,g)$ is timelike geodesically complete. Let $g_\eps :=\check{g}_\eps$ be smooth approximations as constructed in Lemma~\ref{lem:approxnice}. Let $\beta :=\frac{1}{2} (n-1)\min_{\Sigma} g(H,\mathbf{n})>0$, where the minimum exists by compactness of $\Sigma$ and continuity of $H$. Since $g_\eps \prec g$ the hypersurface $\Sigma$ is spacelike for $g_\eps$ as well\footnote{If $g_\eps(X,X)$ were $\leq 0$ for any non-zero $X\in T\Sigma$, then, since $g_\eps \prec g$, $g(X,X)<0$ which contradicts $g$-spacelikeness of $\Sigma$. So $g_\epsilon (X,X)>0$ for all non-zero $X\in T\Sigma$, i.e., $\Sigma$ is $g_\eps$-spacelike.}
 and the $C^1_{\mathrm{loc}}$ convergence of $g_\eps \to g$  implies $\mathbf{n_\eps}\to \mathbf{n}$
 and $H_\eps\to H$ uniformly on $\Sigma$. Thus, for all $\eps$ small enough, we have $(n-1)g_\eps(H_\eps,\mathbf{n_\eps})>\beta>0$ on $\Sigma$. 

Let now $b:= 4\frac{n-1}{\beta}$ and $\delta_0:=\frac{3\beta}{b}\frac{1}{2}=\frac{3\beta^2}{8(n-1)}>0$. Let \[K:=\bigcup_{0<\eps\leq 1}\{v\in TM|_\Sigma: v\perp_\eps T\Sigma,\,g_\eps(v,v)=-1\}\cup \{v\in TM|_\Sigma: v\perp_g T\Sigma,\,g(v,v)=-1\}.\] This is a compact subset of $TM$. So $F_{\leq \eps_0, K,b} \sse TM$ is relatively compact by 
Proposition~\ref{prop:relcompactfunnels}. Let $\tilde{K}\sse TM$ be compact with $F_{\leq \eps_0,K,b}\sse \tilde{K}$ (by Rem.~\ref{rem:compactfunnels} we could actually just take $\tilde{K}=F_{\leq \eps_0,K,b}$). So by Lemma~\ref{lem:thecurvatureapprox} there exists $\eps_1>0$ (depending only on $\tilde{K}$ and $\delta_0$, i.e., $\beta$) such that
\[
\forall \eps<\eps_1\
\forall X \in \tilde{K}  \text{ with }\ g_\eps (X,X)=-1 :
\ \Ric_\eps(X,X) > -\delta_0.\]
 
Fix $\eps <\min\{\eps_1,\eps_0\}$. Then $\dot{\gamma}\in F_{\leq \eps_0,K,b}\sse \tilde{K}$ for all $g_\eps$-unit speed $g_\eps$-geodesics $\gamma:[0,b]\to M$ starting $g_\eps$-orthogonally to $\Sigma$, so for any such geodesic $\Ric(\dot{\gamma},\dot{\gamma})\geq -\delta_0$ and $(n-1)g_\eps (H_\eps,\dot{\gamma}(0))=(n-1)g_\eps (H_\eps,\mathbf{n}_\eps)>\beta$. Thus, Lemma~\ref{lem:hawkingprep} applies and any future directed $g_\eps$-unit speed $g_\eps$-geodesic starting $g_\eps$-orthogonally to $\Sigma$ will stop maximizing the $g_\eps$-distance to $\Sigma $ after $b_0:=2\frac{n-1}{\beta}<b$. 

We now proceed as in the classical proof of Hawking's theorem: This shows that $D^+_\eps(\Sigma)$ is contained in the compact set $F_{\eps, K,b_0}\sse F_{\eps, K,b}$, which implies that the future Cauchy horizon $H^+_\eps (\Sigma)$ is compact and non-empty (since clearly $\emptyset \neq F_{\eps, K,b}\setminus F_{\eps, K,b_0}\subset I^+_\eps(\Sigma)\setminus D^+_\eps(\Sigma)$) and get a contradiction by the usual argument (cf.~e.g, \cite[Theorem~14.55B]{ON83}).
\end{proof}

\begin{remark}	\begin{itemize}
\item[(i)] We can even draw the following more precise conclusion: There exists at least one inextendible future directed timelike geodesic $\gamma $ starting at $\Sigma $ with length less than $\frac{n-1}{\beta}$ for any $\beta<(n-1)\min_{\Sigma} g(H,\mathbf{n})$. This follows, essentially, from restricting to future directed $v$ in the definition of $K$ and replacing Proposition~\ref{prop:relcompactfunnels} with the more precise Rem.~\ref{rem:relcompactfunnels}.
\item[(ii)] Note that this neither uses nor shows that $D^+(\Sigma)\sse \{p\in M:\tau_\Sigma(p)\leq \frac{n-1}{\beta}\}$. This can, however, be inferred from the next theorem.
\end{itemize}
\end{remark}

The next theorem is a direct generalization of \cite[Theorem~14.55A]{ON83}.

\begin{Theorem}[$C^1$-Hawking, version 2]\label{thm:c1hawk2} Let $(M,g)$ be a time-oriented Lorentzian manifold with $g\in C^1$. If $(M,g)$ has non-negative timelike Ricci curvature (in the distributional sense of Def.~\ref{def:sec}) and there exists a smooth spacelike Cauchy hypersurface $\Sigma$ with $(n-1)g(H,\mathbf{n})>\beta >0$ on $\Sigma$, then
$\tau_\Sigma(p)\leq \frac{n-1}{\beta}$ for all $p\in I^+(\Sigma)$.
\end{Theorem}
\begin{proof}
The key additional ingredient for this version is Proposition~\ref{prop:existmaxgeods} or more specifically Corollary~\ref{cor:sigmamaxgeod}. Assume there exists $p\in I^+(\Sigma)$ with $\tau_\Sigma(p)> \frac{n-1}{\beta}$. Let $\check{g}_\eps$ be approximations as in Lemma~\ref{lem:approxnice}, then by Corollary~\ref{cor:sigmamaxgeod} there exists a sequence of $\check{g}_{\eps_k}$-geodesics $\tilde{\gamma}_{\eps_k}:[0,1]\to M$ maximizing the $g_{\eps_k}$-distance from $\Sigma$ to $p$ and converging in $C^1([0,1])$ to a maximizing $g$-geodesic $\tilde{\gamma } : [0,1]\to M$. Let $l_k:=L_{\eps_k}(\tilde{\gamma}_{\eps_k})$ and $l:=L(\tilde{\gamma})=\tau_\Sigma(p)$, then $l_k\to l$. We reparametrize these geodesics to geodesics $\gamma_k:[0,l_k]\to M$ (resp.~$\gamma:[0,l]\to M$) that have $\check{g}_{\eps_k}$ (resp.~$g$) unit speed and note that because of the $C^1([0,1])$ convergence of the original curves 
the set $K:=\bigcup_{k\in \N}\mathrm{im} (\dot{\gamma}_{k})\cup \mathrm{im} (\dot{\gamma})$ is compact. 

Let $U$ be a neighborhood of $\gamma(0)\in \Sigma$ and let $N_0\in \N$ be large enough such that $\gamma_k(0)\in U$ and $(n-1)g_{\eps_k} (H_{\eps_k},\mathbf{n}_{\eps_k})>\beta$ on $U$ for all $k>N_0$. Let $0<c<1$ be such that $\tau_\Sigma(p)>\frac{n-1}{c\beta}>\frac{n-1}{\beta}$ and set $\delta:=\frac{3\beta}{\tau_\Sigma(p)}(1-c)$. Then by Lemma~\ref{lem:thecurvatureapprox} there exists $\eps_0>0$ (depending only on $K$ and $\delta$) such that
\[
\forall \eps<\eps_0\
\forall X \in K  \text{ with }\ g_\eps (X,X)=-1 :
\ \Ric_\eps(X,X) > -\delta.\]
 
Now let $k>N_0$ such that $\eps_k <\eps_0$. Then for the $\check{g}_{\eps_n}$-unit speed $\check{g}_{\eps_k}$-geodesics $\gamma_k:[0,l_k]\to M$ 
we have $\Ric(\dot{\gamma},\dot{\gamma})\geq -\delta$ and $(n-1)g_{\eps_k} (H_{\eps_k},\dot{\gamma}_k(0))=(n-1)g_{\eps_k} (H_{\eps_k},\mathbf{n}_{\eps_k})>\beta$. Thus, Lemma~\ref{lem:hawkingprep} applies and gives $l_k\leq \frac{n-1}{c\beta}$. But then $\limsup_{k\to\infty}l_k\leq \frac{n-1}{c\beta}<\tau_\Sigma(p)$, contradicting $l_k=L_{\eps_k}(\gamma_k)\to L(\gamma)=\tau_\Sigma(p)$.
\end{proof}

\begin{remark}\label{rem:complete} Comparing Theorems \ref{thm:c1hawk1} and \ref{thm:c1hawk2} to their smooth counterparts, namely \cite[Theorem~14.55B]{ON83} and \cite[Theorem~14.55A]{ON83}, one can't help but notice that the order of presentation (and proof) has switched. For people familiar with the smooth proofs this may at first seem a little strange since traditionally \cite[Theorem~14.55A]{ON83} is regarded as the ''easier''/''first'' version and is a crucial tool in the proof of \cite[Theorem~14.55B]{ON83}. What is, intuitively, responsible for switching the ''natural'' order of these theorems around for $C^1$-metrics is that geodesics may be non-unique which makes it harder to prove statements about \emph{every} $g$-geodesic starting orthogonally to $\Sigma$. Or expressed differently: Both Theorem \ref{thm:c1hawk1} and \ref{thm:c1hawk2} are proved using contradiction and assuming that \emph{all} timelike geodesics starting orthogonally to $\Sigma $ are complete (or at least exist for very long times) gives much more leverage than merely assuming that one maximizing curve has long length. 

In particular, as shown in section \ref{sec:geod}, (subsequences of) geodesics for approximating metrics $g_\eps$ will converge to $g$-geodesics, so it is a reasonable expectation that any property which is satisfied by \emph{all} $g$-geodesics will imply something for $g_\eps$ geodesics. And indeed this is key in the proof of Theorem \ref{thm:c1hawk1}. But as discussed in subsection \ref{sec:existofmax}, it is unclear whether any arbitrary maximizing curve or even maximizing geodesic is a limit of  (a subsequence of) maximizing geodesics of approximating metrics.\footnote{And it is easy to imagine that there exist different approximating sequences with different limit geodesics, so at least once a choice of a specific net  $g_\eps$ of approximating metrics is made one may no longer be able to obtain any geodesic as a limit of a subsequence of $g_\eps $ geodesics.}  Luckily, the existence result in Proposition~\ref{prop:existmaxgeods} guarantees that given any maximizing curve between two points $p$ and $q$ there must also exist a (possibly different) maximizing curve from $p$ to $q$ that is exactly such a limit.

In light of there potentially being no relation between timelike geodesics and locally maximizing timelike curves (see the discussion in Remark \ref{rem:geoddef}) one should also note that Theorem \ref{thm:c1hawk1} only shows the existence of an incomplete timelike geodesic while Theorem \ref{thm:c1hawk2} implies both the existence of incomplete timelike geodesics and of incomplete timelike locally maximizing curves. In fact, it implies that any inextendible timelike curve is incomplete (i.e., has finite length), so in particular there exists timelike initial data that doesn't admit any complete geodesic.
\end{remark}

\section{The Penrose singularity theorem}\label{sec:penrose}

To formulate a statement analogous to the Penrose singularity theorem for $C^1$-metrics we have to first discuss how the classical null energy condition, $\Ric(X,X)\geq 0$ for all null vectors $X\in TM$, can be formulated for $C^1$-metrics where $\Ric$ is merely distributional. In the case of timelike Ricci bounds, $\Ric(X,X)\geq 0$ for all timelike vectors $X\in TM$, this generalization was straightforward and we simply required $\Ric(\MX,\MX)\geq 0$ as a distribution for all local smooth timelike vector fields $\MX$. This worked, essentially, because timelikeness is an open condition and every timelike vector $X\in TM$ can be extended to a \emph{smooth} timelike vector field even if the metric is merely $C^1$. And, even further, \emph{any} ''local enough'' smooth extension (i.e., extending to a small enough neighborhood) will be timelike. In particular this is true for any extension which is constant in some chart, a crucial fact used in the proof of both the $C^{1,1}$ and the $C^1$ Hawking singularity theorem. 

However, if $X\in TM$ is null, we will in general only be able to extend $X$ to a \emph{smooth} null vector field if the metric itself is smooth. And, no matter how smooth the metric is and how small of a neighborhood one extends to, there will always be smooth extensions which do not possess any causal character. In particular, we cannot control the causal character of the extension of $X$ to a constant vector field in some given coordinates. For $g\in C^{1,1}$, the workaround for this problem was to work with two different extensions: First, one can locally extend a null vector to a Lipschitz continuous null vector field via parallel transport. And second, local boundedness of $\Ric$ implies that on small enough neighborhoods $\Ric$ of this null extension is arbitrarily close to $\Ric $ of a constant extension in a given chart. Because of this it was sufficient to formulate the null energy condition in \cite[Theorem~1.1]{penrosec11} as
\begin{enumerate}
	\item[(i)] For any Lipschitz-continuous local null vector field $\MX$, $\Ric(\MX,\MX)\geq 0$ almost everywhere.
\end{enumerate}

Now for a $C^1$-metric parallel transport is ill-defined and so this method of extending a null vector to a local null vector field is not available. While it should still be possible to extend any null vector $X$ to a local $C^1$-null vector field $\MX$ using different methods, this still would not help us much because $\Ric$ is no longer locally bounded and so having some estimate for $\Ric(\MX,\MX)$ (which can be properly defined despite $\MX$ only being $C^1$ and not smooth because $\Ric$ is a distribution of order one) doesn't tell us anything about $\Ric (\MX',\MX')$ for another extension $\MX'$ (e.g.~a constant extension in a given chart), no matter how close $\MX$ and $\MX'$ are. 

We therefore give the following definition, essentially stating that if $\MX$ is close to being null, then $\Ric(\MX,\MX)$ has to be almost non-negative:

\begin{definition}[Distributional null energy condition]\label{def:genNEC} We say that $g\in C^1$ satisfies the \emph{distributional null energy condition} if for any compact $K\sse M$ and any $\delta >0$ there exists $\eps(\delta, K)$ such that $\Ric(\MX,\MX)>-\delta$ (as distributions, cf.~Def.~\ref{def:posdistribs}) for any local smooth vector field $\MX\in \mathfrak{X}(U)$ ($U\sse K$) with $||\MX||_h=1$ and satisfying $|g(\MX,\MX)|<\eps(\delta, K)$ on $U$.
\end{definition}

While this does not assume anything about the causal character of the smooth local vector field $\MX$, the condition that $|g(\MX,\MX)|<\eps(\delta, K)$ in this definition is equivalent to $\MX$ being close to a $C^1$-null vector field $\mathcal{N}$.
More precisely, we have the following lemma: 

\begin{Lemma}\label{lem:genNECequiv}
A metric $g\in C^1$ satisfies the distributional null energy condition if and only if for any compact $K\sse M$ and any $\delta >0$ there exists $\eps(\delta, K)$ such that $\Ric(\MX,\MX)>-\delta$ for any (local) smooth vector field $\MX\in \mathfrak{X}(U)$ ($U\sse K$) with $||\MX||_h=1$  and which is $\eps (\delta, K)$ close to a $C^1$ $g$-null vector field $\mathcal{N}$ on $U$, i.e., $||\MX-\mathcal{N}||_h<\eps(\delta, K)$ on $U$.
\end{Lemma}
\begin{proof}
To show that this implies the distributional null energy condition we show that for every compact $K\sse M$ and any $\eps_0(\delta, K) >0$ there exists $\eps(\eps_0, K)>0$ such that for any local smooth causal vector field $\MX\in \mathfrak{X}(U)$ ($U\sse K$) with $||\MX||_h=1$ and $|g(\MX,\MX)|<\eps $ on $U$ there exists a distributional null vector field $\mathcal{N}$ with $||\mathcal{N}-\MX||_h <\eps_0$.

Let $E_0,\dots,E_{n-1}$ be $C^1$ $g$-orthonormal vector fields on $K$ with $E_0$ timelike and $E_1,$ $\dots, E_{n-1}$ spacelike. Then $\MX=\MX^i E_i$, $\MX^i\in C^1(U)$, and $-a\eps=g(\MX,\MX)=-(\MX^0)^2+(\MX^1)^2+\dots+ (\MX^n)^2$, where $a:U\to [0,1)$ is $C^1$. From $||\MX||_h=1$ we get that $||\MX||_{e} \geq C>0$ and hence this implies $|\MX^0|^2\geq c>0$ on $U$ for constants $C,c$ depending on $K$ (and $h$ and the $E_i$) but not on $\MX$. Define 
$\mathcal{N}:=(\sqrt{1-\frac{a\eps}{(\MX^0)^2}} \MX^0,\MX^1,\dots,\MX^{n-1})$ (in the basis of the $E_i$). This is null, $C^1$ and 
\begin{align*}
||\MX-\mathcal{N}||_h&\leq C' ||\MX-\mathcal{N}||_e=C'\,(1-\sqrt{1-\frac{a\eps}{(\MX^0)^2}})||(\MX^0,0,\dots,0)||_{e} \\ &\leq C'\,(1-\sqrt{1-\frac{a}{c}\eps}) ||\MX||_{e} 
\leq C'(1-\sqrt{1-\frac{a}{c}\eps})\to 0 \;\; \text{as}\,\, \eps \to 0, \end{align*}
where $C'$ is a running constant depending only on $K$.

To show the other direction, we have to show that for every compact $K\sse M$ and any $\eps_0(\delta, K) >0$ there exists $\eps(\eps_0, K)>0$ such that any local smooth causal vector field $\MX$ with $||\MX||_h=1$ and $||\MX-\mathcal{N}||_h<\eps$ for some  $C^1$ $g$-null vector field $\mathcal{N}$ satisfies $|g(\MX,\MX)|<\eps_0 $ on $U$. This follows immediately from noting that $|g(\MX,\MX)|= |g(\mathcal{N},\mathcal{N})-g(\MX,\MX)|\leq C ||\mathcal{N}-\MX||_h$ on $K$.
\end{proof}

Note that Definition \ref{def:genNEC} and its equivalent condition in Lemma~\ref{lem:genNECequiv} can also be restated in the following way by rescaling $\MX $:
\begin{definition}[Distributional null energy condition, version 2]\label{def:genNECno2} A Lorentzian metric $g\in C^1$ satisfies the distributional null energy condition if and only if for any compact $K\sse M$, any constants $c_1, c_2>0$ and any $\delta >0$ there exists $\eps(\delta, K, c_i)$ such that $\Ric(\MX,\MX)>-\delta$ (as distributions, cf.~Def.~\ref{def:posdistribs}) for any local smooth vector field $\MX\in \mathfrak{X}(U)$ ($U\sse K$) with $0<c_1\leq ||\MX||_h\leq c_2$ and satisfying $|g(\MX,\MX)|<\eps(\delta, K, c_i)$ on $U$.
\end{definition}

\begin{remark}	[The distributional null energy condition is equivalent to the usual null energy condition if $g$ is of sufficient regularity]
	From the characterization given in Lemma \ref{lem:genNECequiv} and (local) boundedness of $\Ric $ it follows that the distributional null energy condition defined above is indeed equivalent to the $C^{1,1}$ null energy condition 
	as formulated in \cite[Theorem~1.1]{penrosec11}
	if the metric is $C^{1,1}$. Further, as argued in \cite{penrosec11}, if $g\in C^2$ this is equivalent to the usual null energy condition.
	
Also, for $C^{1,1}$-metrics one may omit the \emph{lower} bound on $||\MX||_h$ in Definition \ref{def:genNECno2} because $\Ric \in L^\infty_{\mathrm{loc}}$ and thus $\Ric(\MX,\MX)$ must be close to zero a.e.~for any $\MX$ close to zero.
	In the $C^1$-case, however, definition \ref{def:genNECno2} appears to be strictly weaker than a version where the lower bound $c_1$ is removed: Any such version would immediately imply that all distributions $\Ric_{ij}$ are locally bounded from below for any $i,j$, irrespective of the causality of the corresponding coordinate vector fields $\partial_i,\partial_j$.\\ 
\end{remark}

As for Hawking's theorem, we will need our own $C^1$-version of an important Lemma in the proof of the $C^{1,1}$-Penrose theorem (\cite[Lemma~2.4]{penrosec11}):

\begin{Lemma}\label{lem:thecurvatureapproxPenrose}
Let $M$ be a smooth manifold with a $C^1$-Lorentzian metric $g$. Let $K$ be a compact subset of $M$, $c_1,c_2>0$ and suppose that $(M,g)$ satisfies the distributional null energy condition (in the sense of Def.~\ref{def:genNEC}).
Then for all $ \delta >0$ there exists $\eps_0>0$ (depending on $K$, $\delta$ and $c_1,c_2$) such that
\begin{equation}\label{minePenr}
\forall \eps<\eps_0\
  \forall X \in TM|_K  \text{ with } 0<c_1 \leq ||X||_h\leq c_2 \text{ and } \check{g}_\eps (X,X)=0 :
  \ \Ric[\check{g}_\eps](X,X) > -\delta
\end{equation}
for smooth approximating metrics $\check{g}_\eps$
as in Lemma~\ref{lem:approxnice}.
\end{Lemma}
\begin{proof}
	Completely analogous to the proof of Lemma~\ref{lem:thecurvatureapprox} it suffices to show \eqref{minePenr} for $\Ric  \star_M \rho_{\eps}$ instead of $\Ric[\check{g}_\eps]$. And further, since we are again in a compact set $K\sse M$, where $\Ric\star_M\rho_{\eps}$ is, essentially, just a finite sum of component wise convolutions in charts, it suffices to look at the case where $M=\R^n$.

	Pick any compact $\tilde{K}$ such that $K\sse \tilde{K}^\circ$ and constants $c_1',c_2' $ with $0<c_1'<c_1<c_2<c_2'$. Then there exists $\eps_0 $ such that for any $X\in TM|_K\subseteq \R^n$ with  $|\check{g}_\eps (X,X)|=0$ for some $\eps <\eps_0$ the extension of $X$ to the constant vector field $\MX :=x\mapsto \, X \in \R^n$ on the neighborhood $B_{2\eps_0}(p)\sse \R^n$ (where $p:=\pi(X)\in K$) satisfies:
	\begin{itemize}
		\item[(i)] $\pi(\MX)\in \tilde{K}$
		\item[(ii)] $0<c_1'\leq ||\MX||_h \leq c_2'$ and
		\item[(iii)] $|\check{g}_\eps(p) (X,X)-g(p')(\MX(p'),\MX(p')|<\eps(\delta,\tilde{K},c_i')$ for all $p'\in B_{2\eps_0}(p)$, i.e., $|g(\MX,\MX)|<\eps(\delta,\tilde{K},c_i')$ (where $\eps(\delta,\tilde{K},c_i')$ is as in Def.~\ref{def:genNECno2})
	\end{itemize}
	 Then by the distributional null energy condition we have $\Ric(\MX,\MX)>-\delta$. So $\left(\Ric(\MX,\MX)\right)*\rho_{\eps}>-\delta$ and
	\[(\Ric*\rho_{\eps})(X, X)=(\left(\Ric(\MX,\MX)\right)*\rho_{\eps})(p)>-\delta\]
	and we are done.\qedhere
\end{proof}

The following lemma for smooth metrics is the null version of Lemma~\ref{lem:hawkingprep}.
\begin{Lemma} \label{lem:penroseprep} Let $g$ be smooth. Let $S $ be an $(n-2)$-dimensional spacelike submanifold and $\gamma :[0,b]\to M$, $p:=\gamma(0)\in S$ be a future directed null~geodesic that maximizes the distance to $S $, i.e., $\tau_S(\gamma(s))=0=L(\gamma|_{[0,s]})$. Then $\dot{\gamma}(0)\perp TS$. Further, if $(n-2) g(H,\dot{\gamma}(0))\geq \beta >0$ and $\Ric(\dot{\gamma},\dot{\gamma})\geq -\delta $ with $0\leq \delta \leq \frac{3\beta}{b}(1-c)$ for some $0<c\leq 1$, then $b\leq \frac{n-2}{c \beta}$.
\end{Lemma}
\begin{proof}
	That $\dot{\gamma}(0)\perp TS $ is standard. The second claim is usually only done for  $\Ric(\dot{\gamma},\dot{\gamma})\geq 0$ (cf.~e.g., \cite[Proposition~10.43]{ON83}), but, as in the timelike version Lemma~\ref{lem:hawkingprep}, similar estimates work if $\Ric(\dot{\gamma},\dot{\gamma})\geq -\delta $: We proceed analogously to \cite[Proposition~10.43]{ON83}.
	Let $\{e_i\}_{i=3}^n$ be an ONB of $T_pS $ and let $E_i$ be the vector fields along $\gamma$ obtained by parallel transport of the $e_i$. Then, since there are no focal points along $\gamma $, for any variation with variational vector field $V(s):=(1-\frac{s}{b})E_i(s)$ this yields
	\[ 0\leq \frac{1}{b}- \int_0^b (1-\frac{s}{b})^2 g(R(\dot{\gamma}, E_i)E_i, \dot{\gamma}) ds - g(\mathrm{II}(e_i,e_i), \dot{\gamma}(0)). 
	\]
	Summing over all $i=3,\dots,n$ (and noting that $\Ric(\dot{\gamma},\dot{\gamma})=\sum_{i=3}^ng(R(\dot{\gamma}, E_i)E_i, \dot{\gamma}) $ in this situation) gives
	\[ 0\leq \frac{n-2}{b} - \int_0^b (1-\frac{s}{b})^2 \Ric(\dot{\gamma},\dot{\gamma}) ds- \beta \leq \frac{n-2}{b}+\delta \frac{b}{3} -\beta \leq \frac{n-2}{b}+\beta (1-c)-\beta=\frac{n-2}{b}-c\beta. \qedhere
	\]
\end{proof}

Using this in the same way as we used \ref{lem:hawkingprep} in  the proof of Theorem \ref{thm:c1hawk2} we can now generalize the Penrose theorem, see, e.g., \cite[Theorem~14.61]{ON83}, to the $C^1$-setting. 

\begin{Theorem}[$C^1$-Penrose]\label{thm:c1penrose} Let $(M,g)$ be a time-oriented globally hyperbolic Lorentzian manifold with $g\in C^1$ and non-compact Cauchy hypersurfaces. If $(M,g)$ satisfies the distributional null energy condition (in the sense of Def.~\ref{def:genNEC}) and contains a closed smooth achronal spacelike $(n-2)$-dimensional submanifold $S$ with past pointing timelike mean curvature vector field $H$, 
then it is null geodesically incomplete.
\end{Theorem}
\begin{proof}
Since $H$ is past-pointing timelike, we have $g(H,X)>0$ for all future pointing null vectors $X\in TS^\perp$. Set $K:=\{X\in TS^\perp: 0\leq g(H,X)\leq 2 \}$, then $K\sse TM$ is compact. We now show that $J^+(S)\setminus I^+(S)\sse F_{K,1}$ (where $F_{K,1}$ is as in \eqref{eq:F}): Assume not, and let $p\in (J^+(S)\setminus I^+(S))\setminus F_{K,1}$. By Corollary \ref{cor:Smaxgeod} we have $p=\gamma(1)$ for a null geodesic $\gamma :[0,1]\to M$ with $\dot{\gamma}(0)\perp TS$ which  is the uniform $C^1$-limit of a sequence of $\check{g}_{\eps_k}$-null $\check{g}_{\eps_k}$-geodesics $\gamma_{\eps_k}:[0,1]\to M$ maximizing the $\check{g}_{\eps_k}$-distance from $S$ to $p$. Since $p\notin F_{K,1}$ we have $g(H,\dot{\gamma}(0))>2$, hence for large enough $k$ also $(n-2)\check{g}_{\eps_k}(H_{\eps_k},\dot{\gamma}_{\eps_k}(0))>2(n-2)=:\beta$. By the $C^1$-convergence (and $\dot{\gamma}(s)\neq 0$ for all $s\in [0,1]$) we have $\gamma_{\eps_k}\sse \tilde{K}$ for some compact $\tilde{K}\sse M$ and $0<c_1\leq ||\dot{\gamma}_{\eps_k}||_h \leq c_2$ for some constants $c_1,c_2$. 
Thus Lemma \ref{lem:thecurvatureapproxPenrose} shows that for large enough $k$ we will have $\Ric[\check{g}_{\eps_k}](\dot{\gamma}_{\eps_k},\dot{\gamma}_{\eps_k})\geq -3\beta \frac{1}{4}=:-3\beta (1-c)$. But then by Lemma \ref{lem:penroseprep}, since the $\gamma_{\eps_k}$ are null geodesics maximizing the distance to $S$, $b=1\leq \frac{n-2}{c \beta}=\frac{1}{2c}=\frac{2}{3}$, a contradiction.

Thus $J^+(S)\setminus I^+(S)$ is relatively compact and closed (by global hyperbolicity, see \cite[Proposition~3.1]{Clemens}), hence compact. Further, $J^+(S)\setminus I^+(S)=\partial I^+(S)$ by closedness of $J^+(p)$ and push-up. 
Fix any smooth metric $g'\prec g$ (close enough to $g$ so that it can be time oriented by the same smooth timelike vector field, e.g., $g'=\check{g}_\eps$ from Lemma~\ref{lem:approxnice}). Then any Cauchy hypersurface for $(M,g)$ is also a Cauchy hypersurface for $(M,g')$. So $(M,g')$ has a non-compact Cauchy hypersurface, say $\Sigma$. Further, $I^+(S)$ is 
a $g'$-future set since~$I^+_{g'}(I^+(S))\sse I^+(I^+(S))= I^+(S)$.  Hence its boundary $\partial I^+(S)$ is a compact $g'$-achronal topological hypersurface. Now one can do the usual construction (see \cite[Theorem~14.61]{ON83}) to obtain a homeomorphism between the non-compact Cauchy hypersurface $\Sigma$ for $g'$ and the compact $\partial I^+(S)$, leading to a contradiction.
\end{proof}

\appendix
\section{Appendix: A $C^1$ version of Myers's Theorem}\label{sec:Riem} 
With the methods developed above it is easy to see that also a $C^1$-version of the Myers theorem holds. This is perhaps not very surprising, because it is known that there are
generalizations of Myers's theorem even for metric measure spaces (see
Corollary~2.6 in \cite{sturm_metricMeasureII}), but the $C^1$-version below does not (or at least not immediately) follow from these generalizations. The reason for this is the conceptually very different definition of the Ricci curvature bounds in our setting (i.e., bounds on the distributions $\Ric (\MX,\MX)\in \mathcal{D}'(M)$ for smooth vector fields $\MX$) and for metric measure spaces (see, e.g.,~\cite{sturm_metricMeasureII,LottVillani2009}). While it has been shown that these definitions are equivalent for smooth metrics 
I am not aware of any results relating distributional bounds on Ricci curvature to the metric measure space setting.

As in the smooth case, for any manifold with a continuous Riemannian metric $h$ one can define a continuous distance function $d_h$ via
\[d_h(p,q):=\inf \{L(\gamma):\,\gamma\;\text{is an absolutely continuous curve from p to q}\}\]
and $(M,d_h)$ becomes a locally compact length space \cite[Proposition 4.1]{A:15}. Thus, by the length space version of Hopf-Rinow, the Hopf-Rinow-Cohn-Vossen Theorem, cf.~\cite[Theorem 2.5.28]{BBI}, (metric space) completeness implies the existence of minimizing curves and is equivalent to relative compactness of bounded sets. If $h$ is $C^1$ a classical result using the trick of Du Bois-Reymond is that any length minimizing curve parametrized with respect to arc-length must be a solution of the geodesic equation (see, e.g., \cite{SS:18}). Thus metric space completeness is equivalent to geodesic completeness: Any geodesic (i.e., solution of the geodesic equation) of finite length must be contained in a bounded set. Hence, by metric space completeness, it is also contained in a compact set and thus (after reparametrizing to unit speed, if necessary) $\dot{\gamma}$ is contained in a compact subset of $TM$, so the solution of the geodesic equation is extendible. The converse is a direct consequence of \cite[Theorem 2.5.28]{BBI} and shortest paths being (reparametrizations of) geodesics.

\begin{Theorem}\label{thm:myers}
	Let $(M,h)$ be a Riemannian manifold with a $C^1$-metric $h$. Assume that $h$ is complete and that there exists a constant $\lambda>0$ such that $\Ric (\MX,\MX)-(n-1)\lambda h(\MX,\MX)$ is a non-negative distribution for all smooth vector fields $\MX$ on $M$. Then $\mathrm{diam}(M)\leq \frac{\pi}{\sqrt{\lambda}}$.
\end{Theorem}
\begin{proof}
	Assume $\mathrm{diam}(M)> \frac{\pi}{\sqrt{\lambda}}$ and choose $\lambda_0<\lambda_1<\lambda$ such that $\mathrm{diam}(M)> \frac{\pi}{\sqrt{\lambda_0}}$. Note that $\Ric (\MX,\MX)-(n-1)\lambda h(\MX,\MX)$ being a non-negative distribution for all smooth vector fields $\MX$ implies that $\Ric(\MX,\MX)>(n-1)\lambda_1$ for all smooth vector fields $\MX$ with $h(\MX,\MX)>\frac{\lambda_1}{\lambda}$. 
	
	Note that the following version of Lemma~\ref{lem:thecurvatureapprox} holds: 
	\begin{Lemma}\label{lem:curvapproxRiem}
	Let $K\sse M$ be compact, $0<c<1$, $\lambda_1>0$ and assume that $\Ric(\MX,\MX)>(n-1)\lambda_1$ for all smooth vector fields $\MX$ with $h(\MX,\MX)> c$. Let $h_\eps$ be smooth approximating Riemannian metrics converging to $h$ in $C^1_\mathrm{loc}$ such that $h_\eps=h\star_M \rho_{\eps}$ on $K$ for $\eps$ small. Then for all $\delta >0$ there exists $\eps_0>0$ such that
	\begin{equation}\label{mineriem}\forall \eps<\eps_0\
	\forall X\in TM|_K \text{ with }\ h_\eps (X,X)=1 :
	\ \Ric[h_\eps](X,X) > (n-1)(\lambda_1-\delta).
	\end{equation}
	\end{Lemma}
\begin{proof}
	Because $h_\eps\to h$ locally uniformly it suffices to show $\Ric[h_\eps](X,X) > (n-1)(\lambda_1-\delta)$ for all $X\in TM|_K$ with $h(X,X)>c$. As in the proof of Lemma~\ref{lem:thecurvatureapprox} and Lemma~\ref{lem:thecurvatureapproxPenrose} we have $\Ric[h_\eps]-\Ric \star_M\rho_{\eps} \to 0$ locally uniformly. So, again, it suffices to consider the case $M=\R^n$ and the proof reduces to showing that $(\Ric *\rho_{\eps})(X,X) >(n-1)\lambda_1$. Extending $X\in T_pM=\R^n$ to the constant vector field $\MX : x\mapsto X$ on a small neighborhood of $p$ preserves $h(\MX,\MX)>c$ and the left hand side becomes $(\Ric(\MX,\MX) *\rho_{\eps})(p)$ and the desired inequality follows immediately from properties of smoothing by convolution and the assumption of $\Ric(\MX,\MX) >(n-1)\lambda_1$.
\end{proof}
Next we need to show that there exist smooth approximations $h_\eps$ such that, first, for any compact $K\sse M$ one has $h_\eps=h\star_M\rho_\eps$ on $K$ for all $\eps$ small enough, and second, $h_\eps$ is complete. We start by showing the existence of approximations satisfying the first point and  $\sup \{|1-h_\eps(X,X)| \,:\, X\in TM,\, |X|_h=1\}\leq \frac{1}{2}$. First observe that by locally uniform convergence for any compact $K\sse M$ there exists $\eps_K$ such that $\sup \{|1-h_\eps(X,X)| \,:\, X\in TM|_K,\, |X|_h=1\}\leq \frac{1}{2}$ for all $\eps\leq \eps_K$. Thus (cf.~e.g.~\cite[ Lemma~4.3]{HKS:12}) we can construct a smooth map $u:(\eps,p)\mapsto \R_{>0}$ such
that for each $\eps$, the metrics $h_\eps : M\to T^0_2M$ defined by $h_\eps (p):=(h \star_M\rho_{u(\eps,p)})(p)$  have the desired properties.

Now since $\sup \{|1-h_\eps(X,X)| \,:\, X\in TM,\, |X|_h=1\}\leq \frac{1}{2} $ all metrics $h_\eps$ will be Riemannian and by Hopf-Rinow it suffices to show relative compactness of $B_\eps(p,r)$ for all $r>0$ and $p\in M$ to show completeness of $h_\eps$. This follows immediately from $B_\eps(p,r)\sse B(p,2r)$ (note that $L_h(\gamma)\leq 2 L_\eps(\gamma)$ for any curve $\gamma$, since $|X|_{h}=1\implies \frac{1}{2} \leq |X|_{h_\eps}\leq 2$).

 Thus we have constructed approximating smooth metrics $h_\eps$ that are complete and for which Lemma~\ref{lem:curvapproxRiem} holds. Now let $p,q\in M$ such that $d_h(p,q)\geq \frac{\pi}{\sqrt{\lambda_0}}$ and let $\gamma$ be a length minimizing $h$-unit speed curve from $p$ to $q$. Then by completeness of $h_\eps$ there exist $h_\eps$-minimizing $h_\eps$-unit speed $h_\eps$-geodesics $\gamma_{{\eps}}$ from $p$ to $q$. 
 We have $d_\eps(p,q)=L_\eps(\gamma_{\eps})\leq L_\eps(\gamma)\leq 2L_h(\gamma)=2 d_h(p,q)$. Since $L_h(\gamma_{{\eps}})\leq 2L_\eps(\gamma_\eps)=2 d_\eps(p,q)$ we obtain $\gamma_{{\eps}}\sse B(p,4d_h(p,q))\subseteq K\sse M$. 
 Now let $\eps$ be small enough such that for this $K$ the approximations satisfy $\Ric[h_\eps] > (n-1)(\lambda_1-\delta)g>(n-1)\lambda_0g$. Then by the proof of Myers's Theorem (cf. \cite[Theorem~9.3.1]{DoCarmo}) we have $L_\eps(\gamma_{{\eps}})\leq \frac{\pi}{\sqrt{\lambda_1-\delta}}<\frac{\pi}{\sqrt{\lambda_0}}$. But this would imply that for small enough $\eps$ also $L_h(\gamma_{{\eps}})<\frac{\pi}{\sqrt{\lambda_0}}$, a contradiction to $d_h(p,q)\geq \frac{\pi}{\sqrt{\lambda_0}}$.
\end{proof}

\medskip
\noindent{\em Acknowledgements.} This work was partially supported by project P28770 of the Austrian Science Fund FWF. Major parts of this work were carried out while the author was at the University of Tübingen. I would like to thank Michael Kunzinger, Clemens S\"amann and Roland Steinbauer for helpful discussions. 


\begin{thebibliography}{00}
	
\bibitem{AGKS:19} Alexander, S.~B., Graf, M., Kunzinger, M., S\"amann, C., Generalized cones as Lorentzian length spaces:
Causality, curvature, and singularity theorems. ArXiv: 1909.09575

\bibitem{BEStab} Beem, J.~K., Ehrlich, P., Geodesic completeness and stability, {\em Math. Proc. Camb. Phil. Soc.} 102, 319--328, 1987.


\bibitem{BBI} Burago, D., Burago, Y., Ivanov S., {\em A course in metric geometry.} American
Mathematical Society, Providence, RI, 2001.

\bibitem{A:15} Burtscher, A., Length structures on manifolds with continuous Riemannian metrics, {\em New York J. Math.} 21, 273--296, 2015.

\bibitem{BLF} Burtscher, A., LeFloch, P. G., The formation of trapped surfaces in spherically-symmetric Einstein-Euler spacetimes with bounded variation. {\em 	J. Math. Pures Appl. 102} (6), 1164--1217, 2014.


\bibitem{CdLS:12} Conti S., De Lellis C., Sz\'{e}kelyhidi L., $h$-Principle and Rigidity for $C^{1,\alpha}$ Isometric Embeddings. In: {\em Nonlinear Partial Differential Equations: The Abel Symposium 2010}, Springer, 2012.

\bibitem{christ} Christodoulou, D., {\em The formation of black holes in general relativity.}  European Mathematical Society, 2009.


\bibitem{CG} Chru\'sciel, P.T., Grant, J.D.E., On Lorentzian causality with continuous metrics. {\em Classical Quantum Gravity} 29(14) 145001, 32 pp., 2012.


\bibitem{deRham} de Rham, G., {\em Differentiable manifolds}. Springer, Berlin, 1984.

\bibitem{DieudIII} Dieudonne, J., {\em Treatise on Analysis. III}. Academic Press, New York, 1972.

\bibitem{DoCarmo} do Carmo, M. P., {\em Riemannian Geometry}. Birkhäuser, Basel, 1992.


\bibitem{Filippov} Filippov, A.~F., {\em Differential Equations with Discontinuous Right-Hand Side.} Kluwer Academic Publishers, 1988.



\bibitem{GT87} Geroch, R., Traschen, J., Strings and other distributional sources in general relativity.
{\em Phys.~Rev.~D} 36 (4), 1017--1031, 1987.


\bibitem{GGKS:18} Graf, M., Grant, J.D.E., Kunzinger M., Steinbauer, R., The Hawking--Penrose singularity theorem for $C^{1,1}$-Lorentzian metrics. {\em Commun. Math. Phys.} 360, 1009--1042, 2018.

\bibitem{GL:18} Graf, M., Ling, E., Maximizers in Lipschitz spacetimes are either timelike or null, {\em Classical and Quantum Gravity} 35, 087001, 2018.

\bibitem{GKSS:19} Grant, J.D.E., Kunzinger, M., Sämann, C., Steinbauer, R., The future is not always open. {\em Lett. Math. Phys.} 110, 83--103, 2020.

\bibitem{TheBook_GKOS} Grosser, M., Kunzinger, M., Oberguggenberger, M., Steinbauer, R., {\em Geometric Theory of  Generalized Functions with Applications to General Relativity.} Springer, 2001.

\bibitem{Hartman} Hartman, P., {\em Ordinary Differential Equations. Second edition.} Society for Industrial and Applied Mathematics, Philadelphia, 2002.

\bibitem{HW:51} Hartman, P., Wintner, A., On the problems of geodesics in the small. {\em Amer.~J.~Math.} 73, 132--148, 1951.

\bibitem{H:67}
S.~W. Hawking, The occurrence of singularities in cosmology. {III}.
	{C}ausality and singularities. {\em Proceedings of the Royal Society of London.
		Series A, Mathematical and Physical Sciences} 300, no.~1461,
	187--201, 1967.

\bibitem{HE} Hawking, S.W., Ellis, G.F.R., {\em The large scale structure of
space-time.} Cambridge University Press, 1973.

\bibitem{HKS:12} H\"ormann, G., Kunzinger, M., Steinbauer, R., Wave equations on non-smooth spacetimes. In: {\em Evolution equations of hyperbolic and Schrödinger type}, Birkh\"auser/Springer, 2012.

\bibitem{HM:19} Hounnonkpe, R. A., Minguzzi, E., Globally hyperbolic spacetimes can be defined without the 'causal' condition. {\em Classical and Quantum Gravity} 36 (19), 197001, 2019.

\bibitem{GKS_main} Kunzinger, M., Sämann, C.,	Lorentzian length spaces. {\em Ann. Glob. Anal. Geom.} 54, no.~3, 399-447, 2018. 

\bibitem{KSS} Kunzinger, M., Steinbauer, R., Stojkovi\'c, M., The exponential
map of a $C^{1,1}$-metric. {\em Differential Geom. Appl.} 34, 14--24, 2014.



\bibitem{KSSV} Kunzinger, M., Steinbauer, R., Stojkovi\'c, M., Vickers, J.A.,
A regularisation approach to causality theory for $C^{1,1}$-Lorentzian metrics.
{\em Gen.\ Relativ.\ Gravit.} 46:1738, 18 pp.\ 2014.

\bibitem{hawkingc11} Kunzinger, M., Steinbauer, R., Stojkovi\'c, M., Vickers, J.A.,
Hawking's singularity theorem for $C^{1,1}$-metrics. {\em Classical Quantum Gravity} 32, 075012, 19pp, 2015

\bibitem{penrosec11} Kunzinger, M., Steinbauer, R., Vickers, J.A.,
The Penrose singularity theorem in regularity $C^{1,1}$. {\em Classical Quantum Gravity 32}, 155010, 12pp, 2015.


\bibitem{LeFlochMadare} LeFloch, P.~G., Mardare, C., Definition and stability of Lorentzian manifolds with distributional curvature. {\em Port. Math.} 64 (4), 535--573, 2007.


\bibitem{LottVillani2009}  Lott, J., Villani, C., Ricci curvature for metric-measure spaces via optimal transport. {\em Annals of
Mathematics 169},  903--991, 2009.

\bibitem{MS}  Mars, M., Senovilla, J.M.M., Geometry of general hypersurfaces in spacetime: junction
conditions. {\em Classical Quantum Gravity 10}, 1865--1897, 1993.

\bibitem{Ming} Minguzzi, E., Convex neighborhoods for {L}ipschitz connections and sprays.
 {\em Monatsh. Math.} 177, no.~4, 569--625, 2015.
 
 \bibitem{Ming:19} Minguzzi., E., Causality theory for closed cone structures with applications. {\em Rev. Math. Phys.}, 31, 1930001, 2019.

\bibitem{MS:19} Minguzzi, E., Suhr, S., Some regularity results for Lorentz--Finsler spaces. {\em Ann. Glob. Anal. Geom.} 56,597--611, 2019. 

\bibitem{ON83} O'Neill, B., {\em Semi-Riemannian Geometry. With Applications to
Relativity.}\ Pure and Applied Mathematics 103. Academic Press, New York, 1983.

\bibitem{OS} Oppenheimer, J.R., Snyder, H., On continued gravitational contraction. {\em Phys. Rev. 56},
455--459, 1939.

\bibitem{Pen} Penrose, R., Gravitational collapse and space-time singularities.
{\em Phys.\ Rev.\ Lett.} 14, 57--59, 1965.

\bibitem{Rendall} Rendall, A.~D., Theorems on existence and global dynamics for the Einstein equations. {\em Living Reviews in Relativity}, 8(6), 2005.

\bibitem{Clemens} Sämann, C., Global Hyperbolicity for Spacetimes with Continuous Metrics. {\em Ann. Henri Poincar\'{e}} 17, 1429--1455, 2017.

\bibitem{SS:18} S\"amann, C., Steinbauer, R., On geodesics in low regularity. \emph{Journal of Physics:} Conferences Series 968 012010, 2018.

\bibitem{Sbierski}
Sbierski, J., The $C^0$-inextendibility of the Schwarzschild spacetime and the spacelike diameter in Lorentzian geometry. {\em J. Differential Geom.} 108, 319--378, 2018.

\bibitem{S:14} Steinbauer, R.,  Every Lipschitz metric has $C^1$-geodesics. {\em Classical Quantum Gravity}, 31, 057001, 2014.
\bibitem{SV09} Steinbauer, R., Vickers, J.~A., On the Geroch-Traschen class of metrics. {\em Class.~Quantum Grav.} 26, 065001, 2009.

\bibitem{sturm_metricMeasureII}  Sturm, K.~T., On the geometry of metric measure spaces. II. {\em Acta Mathematica 196}, 133--
177, 2006.

\end{thebibliography}
\end{document}